\pdfoutput=1

\newif\iflongversion

\longversiontrue % or
\iflongversion
\documentclass[12pt,draftcls,onecolumn]{IEEEtran}
\else % short version
\documentclass{IEEEtran}
\fi

\usepackage{amsmath,amsthm,amssymb,array,graphicx,verbatim,color}
\usepackage{enumerate}
\usepackage[numbers, sort&compress]{natbib}
\usepackage{algorithm}
\usepackage{algcompatible}

\newtheorem{theorem}{Theorem}
\newtheorem{lemma}{Lemma}
\newtheorem{definition}{Definition}

\newtheorem{remark}{Remark}

\newcommand{\CI}{CIB}

\DeclareMathOperator*{\argmax}{arg\,max\,}
\usepackage{epstopdf}
\usepackage{mathtools}

\begin{document}

\title{
%On a Class of Dynamic Games with Asymmetric Information and Private Markovian Dynamics
%
%On a Class of Stochastic Dynamic Games with Asymmetric Information: Common Information Based Equilibria and Sequential Decomposition
Dynamic Games with Asymmetric Information: Common Information Based Perfect Bayesian Equilibria and Sequential Decomposition
}

\author{Yi Ouyang, Hamidreza Tavafoghi and Demosthenis Teneketzis% <-this % stops a space
\thanks{This research was supported in part by NSF grants CCF-1111061 and CNS-1238962.
\newline\hspace*{5pt}
Y. Ouyang, H. Tavafoghi and D. Teneketzis are with the Department of Electrical Engineering and Computer Science, University of Michigan, Ann Arbor, MI (e-mail: ouyangyi@umich.edu; tavaf@umich.edu; teneket@umich.edu).
\newline\hspace*{5pt}
A preliminary version of this paper will appear in \textit{the Proceeding of the 54th IEEE Conference on Decision and Control (CDC)}, Osaka, Japan, December 2015.
}
}

\date{\today}
\maketitle
\begin{abstract}
We formulate and analyze a general class of stochastic dynamic games with asymmetric information arising in dynamic systems. In such games, multiple strategic agents control the system dynamics and have different information about the system over time.
Because of the presence of asymmetric information, each agent needs to form beliefs about other agents' private information. 
Therefore, the specification of the agents' beliefs along with their strategies is necessary to study the dynamic game.
We use Perfect Bayesian equilibrium (PBE) as our solution concept. A PBE consists of a pair of strategy profile and belief system.
In a PBE, every agent's strategy should be a best response under the belief system, and the belief system depends on agents' strategy profile when there is signaling among agents.
Therefore, the circular dependence between strategy profile and belief system makes it difficult to compute PBE. 
Using the common information among agents, we introduce a subclass of PBE called common information based perfect Bayesian equilibria (\CI-PBE), and provide a sequential decomposition of the dynamic game. Such decomposition leads to a backward induction algorithm to compute \CI-PBE. We illustrate the sequential decomposition with an example of a multiple access broadcast game. We prove the existence  of \CI-PBE for a subclass of dynamic games.
\end{abstract}


\section{Introduction}
\label{sec:intro}
\subsection*{Background and Motivation}

Stochastic dynamic games arise in many socio-technological systems such as cyber-security systems, electronic commerce platforms, communication networks, etc.
In all these systems, there are many strategic decision makers (agents). %At each time, each agent makes decision based on his whole history that includes his observation and his past decisions so as to maximize his total payoff over time, estimating (anticipating) the strategy and the information of the other agents in the past (the present, and the future). 
In dynamic games with symmetric information all the agents share the same information and each agent makes decisions anticipating other agents' strategies. This class of dynamic games has been extensively studied in the literature (see \cite{fudenberg1991game,osborne1994course,basar1995dynamic,myerson2013game,filar1996competitive} and references therein). An appropriate solution concept for this class of games is sub-game perfect equilibrium (SPE), which consists of a strategy profile of agents that must satisfy \textit{sequential rationality} \cite{fudenberg1991game,osborne1994course}. The common history in dynamic games with symmetric information can be utilized to provide a sequential decomposition of the dynamic game. The common history (or a function of it) serves as an information state and SPE can be computed through backward induction.

Many instances of stochastic dynamic games involve asymmetric information, that is, 
agents have different information over time (such games are also called dynamic games of incomplete information in the game and economic theory literature).
In communication networks, different nodes have access to different local observations of the network.
In electronic commerce systems, each seller has private information about the quality of his product.
In cyber-security systems, a defender can not directly detect the attacker's activities. 
In this situation, if an agent wants to assess the performance of any particular strategy, he needs to form beliefs (over time) about the other agents' private information that is relevant to his objective.
%In dynamic games with asymmetric information, different agents have different local observations. 
%As a result, each agent needs to anticipate the other agents' strategy and to form a belief about the other agents' local observations. 
Therefore, \textit{perfect Bayesian equilibrium} (PBE) is an appropriate solution concept for this class of games. PBE consists of a pair of strategy profile and belief system for all agents that jointly must satisfy \textit{sequential rationality} and \textit{consistency} \cite{fudenberg1991game,osborne1994course}. In games with asymmetric information a decomposition similar to that of games with symmetric information is not possible in general. 
This is because the evaluation of an agent's strategy depends, in general, on the agent's beliefs about all other agents' private information over time. Since private information increases with time, the space of beliefs on the agents' private information grows with time.
%This is because the computation of any belief at time $t$ depends, in general, on the strategy of all agents up to time $t-1$. 
As a result, sequential computation of equilibria for stochastic dynamic games with asymmetric information is available only for special instances (see \cite{zamir1992repeated,forges1992repeated,aumann1995repeated,mailath2006repeated,jones2012policy,renault2006value,gensbittel2012value,renault2012value,cardaliaguet2015markov,nayyar2014Common,gupta2014common} and references therein). 

%Only special instances of stochastic dynamic games with asymmetric information have been studied in the literature (see \cite{nayyar2014Common,gupta2014common,jones2012policy,renault2006value,gensbittel2012value,renault2012value,mailath2006repeated} and references therein). 

In this paper, we consider a general model of a dynamic game with a finite number of agents/players in a system with asymmetric information.
The information available to an agent at any time can be decomposed into \textit{common information} and \textit{private information}. Common information refers to the part of an agent's information that is known by all agents; private information includes the part of an agent's information that is known only by that agent. 
%\red{The presence of private information makes the agents' information asymmetric.}
We define a class of PBE and provide a sequential decomposition of the game through an appropriate choice of information state using ideas from the common information approach for decentralized decision-making, developed in \cite{nayyar2013decentralized}. The proposed equilibrium and the associated decomposition resemble Markov perfect equilibrium (MPE), defined in \cite{maskin2001markov} for dynamic games with symmetric information.

Games with asymmetric information have been investigated in the economic literature within the context of repeated games of incomplete information (see \cite{zamir1992repeated,forges1992repeated,aumann1995repeated,mailath2006repeated} and references therein).
A key feature of these games is the absence of dynamics.
The problems investigated in \cite{nayyar2014Common,gupta2014common,jones2012policy,renault2006value,gensbittel2012value,renault2012value,cardaliaguet2015markov} are the most closely related to our problem. The authors of \cite{jones2012policy,renault2006value,gensbittel2012value,renault2012value,cardaliaguet2015markov} analyze zero-sum games with asymmetric information. The authors of \cite{nayyar2014Common,gupta2014common} used a common information based methodology, inspired by \cite{nayyar2013decentralized}, to establish the concept of \textit{common information based Markov perfect equilibria}, and to achieve a sequential decomposition of the dynamic game that leads to a backward induction algorithm that determines such equilibria.
Our problem is different from those investigated in \cite{nayyar2014Common,gupta2014common,jones2012policy,renault2006value,gensbittel2012value,renault2012value,cardaliaguet2015markov} 
for the following reasons. It is a nonzero-sum game, thus, it is different from the problems analyzed in \cite{jones2012policy,renault2006value,gensbittel2012value,renault2012value,cardaliaguet2015markov}.
Our approach to analyzing dynamic games with asymmetric information is similar to that of \cite{nayyar2014Common,gupta2014common}; the key difference between our problem and those in \cite{nayyar2014Common,gupta2014common} is in the information structure. The information structure in \cite{nayyar2014Common,gupta2014common} is such that the agents' common information based (\CI) beliefs are \textit{strategy-independent}, therefore there is \textit{no signaling} effect. This naturally leads to the concept of common information based Markov perfect equilibrium. 
In our problem the information structure is such that the \CI\ beliefs are \textit{strategy-dependent}, thus \textit{signaling} is present.  
In such a case, the specification of a belief system along with a strategy profile is necessary to analyze the dynamic game.
Signaling is a key phenomenon present in stochastic dynamic games with asymmetric information. Since it plays a fundamental role in the class of games we investigate in this paper, we discuss its nature and its role below.
The discussion will allow us to clarify the nature of our problem, after we formulate it, and to contrast it with the existing literature, in particular \cite{nayyar2014Common,gupta2014common}.

% to establish the concept of common information-based equilibria, and achieve a sequential decomposition of the dynamic game that leads to a backward induction algorithm that determines such equilibria. The key difference between our problem and those in \cite{nayyar2014Common,gupta2014common} is in the information structure. The information structure in \cite{nayyar2014Common,gupta2014common} is such that the agents' common information based belief about the system state is policy independent. Therefore, in \cite{nayyar2014Common,gupta2014common} there is no \textit{signaling} effect as discussed below.

\subsection*{Signaling}

%In a dynamic game with asymmetric information, the information available to an agent at any time can be decomposed into \textit{common information} and \textit{private information}. Common information refers to the part of information that is known by all agents, and private information includes the part of information that is known by the agent, but not known by all other agents at the time. The presence of private information makes the agents' information asymmetric.

In a dynamic game with asymmetric information, an agent's private information is not observed directly by other agents. 
Nevertheless, when an agent's strategy depends on his private information, part of this private information may be revealed/transmitted through his actions. 
We call such a strategy a \textit{private strategy}.
%depending on the agent's private information a private strategy.
When the revealed information from an agent's private strategy is ``\textit{relevant}'' to other agents, the other agents utilize this information to make future decisions.
This phenomenon is referred to as \textit{signaling} in games \cite{kreps1994signalling} and in decentralized control \cite{ho1980team}. 
When signaling occurs, agents' beliefs about the system's private information (which is defined to be the union of all agents' private information) depend on the agents' strategies (see \cite{kreps1994signalling}).
Signaling may occur in games with asymmetric information depending on the system dynamics, the agents' utilities and the information structure of the game.
Below we identify game environments where signaling occurs, as well as environments where signaling does not occur.

To identify game environments where signaling occurs we need to precisely define what we mean by the statement: an agent's private information is ``\textit{relevant}'' to other agents. For that matter we define the concepts of payoff relevant and payoff irrelevant information.

%In general, signaling exists in games with asymmetric information. We call an agent's strategy a signaling strategy if the strategy depends on the agent's private information. That is, the agent will take different actions based on different realizations of his private information under a signaling strategy.
%When an agent uses a signaling strategy, other agents' beliefs about his private information depends on the agent's strategy.
%However, the appearance of information signaling/transmission from one agent to other agents depends on the structure of agents' utilities and system dynamics. There may be no signaling possibility in some situations even when an agent uses a signaling strategy.
%
%
%We discuss in the following the signaling possibility for four different situations based on how an agent's private information affects agents' utilities. 

We call a variable (e.g. the system state, an observation, or an action) \textit{payoff relevant} (respectively, \textit{payoff irrelevant}) to an agent at time $t$ if the agent's expected continuation utility at $t$ directly depends on (respectively, does not depend on) this variable given any fixed realization of all other variables\footnote{Decomposition of agents' types to payoff-relevant type and payoff-irrelevant type is a standard decomposition in the economic literature. Here, we use term 'variable' instead of 'type' to match with the existing literature in control theory. For a more rigorous definition consult with \cite[ch.9]{borgers2015introduction}.}. For instance, in a dynamic game with Markov dynamics where agents' utilities at each time only depend on the current states, the current states are payoff relevant and the history of previous states is payoff irrelevant.

%When agents' actions are all payoff relevant to all agents (i.e. there is no action purely for communication \cite{}), 
%There are four situations according to the payoff relevancy of an agent's private information.
There are four types of game environments depending on the payoff revelance of an agent's private information.

\begin{enumerate}[(a)]
\item Agent $n$'s private information at $t$ is payoff relevant to him at $t$ and from $t+1$ on, but payoff irrelevant to other agents from $t+1$ on. 
%In this game environment, agent $n$ may use a signaling strategy to select different actions for different outcomes of his private information, because the private information is payoff relevant to him at $t$. 
In this game environment, agent $n$ may use a private strategy at $t$ because his private information is payoff relevant to him at $t$. 
Then, other agents can infer part of agent $n$'s privation information at $t$ based on agent $n$'s action. Although this revealed private information is payoff irrelevant to other agents, they can use it to anticipate agent $n$'s future actions since this information is payoff relevant to agent $n$'s future utility.
In this game environment, signaling from agent $n$ to other agents occurs.
%since other agents' decisions depends on their belief about agent $n$'s private information at $t$, and this belief is formed based on agent $n$'s strategy.

\item Agent $n$'s private information at $t$  is payoff irrelevant to him at $t$, and is payoff relevant to other agents from $t+1$ on. 
This class of games includes the classic \textit{cheap-talk game} \cite{crawford1982strategic}. In this game environment, other agents form beliefs about agent $n$'s private information at $t$ because it is payoff relevant to them. 
By using a private strategy, agent $n$ can affect other agents' beliefs about his private information, hence, affect other agents' future decisions. Signaling may occur in this situation if agent $n$ can improve his future utility when he \textit{signals} part of his private information through his actions (e.g. perfectly informative/separating equilibria in the cheap-talk game). 
There may be no signaling if by revealing part of his private information agent $n$ does not increase his future utility (e.g. uninformative/pooling equilibria in the cheap-talk game).

\item Agent $n$'s private information at $t$ is payoff relevant to him at $t$, and payoff relevant to other agents from $t+1$ on.  
This game environment has both effects discussed in the previous two environments. As a result, we may have signaling or no signaling from agent $n$, depending on whether or not he can improve his future utility by using a private strategy. 
Decentralized team problems are examples where signaling occurs, because signaling strategies can help the collaborating agents to achieve higher utilities (see \cite{nayyar2011sequential,nayyar2015signaling,ouyang2015signaling,ouyang2015common} for examples of signaling strategies in decentralized team problems).
Pooling equilibria in the classic two-step \textit{signaling game}  \cite{kreps1994signalling} is an example of no signaling.

%In this situation, we have both effects discussed in the two previous situations. As a result, we may have signaling or no signaling from agent $n$ to other agents depending on whether agent $n$ can improve his future utility from a signaling strategy. 
%Decentralized team problems are examples where signaling happen because signaling strategies can help the collaborative agents to achieve higher utilities (see \cite{nayyar2011sequential,ouyang2015signaling} for examples of signaling strategies in decentralized team problems).
%Pooling equilibria in the classic two-step \textit{signaling game}  \cite{kreps1994signalling} is an example of no signaling.

\item Agent $n$'s private information at $t$ is payoff irrelevant to all agents, including himself, from $t+1$ on. 
In this game environment no signaling occurs. Even if agent $n$ uses a private strategy at $t$, since his private information is payoff irrelevant from $t+1$ on to all agents, no agent will incorporate it in their future decisions
% when his private information is payoff relevant at $t$, no agent will incorporate it in their future decisions since this private information is payoff irrelevant to all agents' future utilities
 \footnote{If one of the agents incorporates the belief on this private information in his strategy from $t+1$ on, all other agents may also incorporate it. The argument is similar to situation (b) since all other agents will anticipate about how this agent will act. 
We note that, agents can use such payoff irrelevant information as a coordination instrument, and therefore, expand their strategy space thereby resulting in additional equilibria. As an example, consider a repeated prisoner's dilemma game with imperfect public monitoring of actions \cite[ch.	7]{mailath2006repeated}. The agents can form a punishment mechanism that results in new equilibria in addition to the repetition of the stage-game equilibrium. 
In general, the idea of such a punishment mechanism is used to proof different versions of folk theorem for different setups \cite{mailath2006repeated}.
However, we do not call this kind of cases signaling because the signals or actions of an agent serves only as a coordination instrument instead of transmitting private information form one agent to other agents.
}. The model in \cite{nayyar2014Common,gupta2014common} are examples of this situation where signaling of information does not occur.

\end{enumerate}

When signaling occurs in a game, all agents' beliefs on the system's private information
are strategy dependent.
Furthermore, each agent's choice of (private) strategy is based on the above mentioned beliefs, as they allow him to evaluate the strategy's performance. This circular dependence between strategies and beliefs makes the computation of equilibria for dynamic games a challenging problem when signaling occurs. This is not the case for games with no signaling effects. In these games, the agents' beliefs are strategy-independent and the circular dependence between strategies and belief breaks. Then, one can directly determine the agents' beliefs first, and then, compute the equilibrium strategies via backward induction \cite{nayyar2014Common,gupta2014common}.

%
%
%
%\blue{
%\subsection*{Sequential Decomposition}
%\begin{itemize}
%\item Talk about the idea of sequential decomposition. With sequential decomposition, we may be able to compute the equilibria of multi-stage dynamic games.
%\item Doing backward induction directly requires the search of all possible beliefs over an information set. The size of information sets generally grows with time.
%\item In a idea sequential decomposition, each step of the decomposition requires finding an equilibrium in a set of a fixed (does not depend on time) beliefs and strategies.
%\end{itemize}
%}

\subsection*{Contribution}

The key contributions of the paper are: (1) The introduction of a subclass of PBE called \textit{common information based perfect Bayesian equilibria} (\CI-PBE) for dynamic games with asymmetric information. A \CI-PBE consists of a pair of strategy profile and a belief system that are sequentially rational and consistent. 
(2) The sequential decomposition of stochastic dynamic games with the asymmetric information through an appropriate choice of information state. This decomposition provides a backward induction algorithm to find \CI-PBE for dynamic games where signaling occurs. The decomposition and the algorithm are illustrated by an example from multiple access communication. 
(3) The existence of \CI-PBE for a subclass of stochastic dynamic games with asymmetric information.

\subsection*{Organization}
The paper is organized as follows. We introduce the model of dynamic games in Section \ref{sec:model}. In Section \ref{sec:solconcept}, we define the solution concept for our model and compare it with that for the standard extensive game form.
In Section \ref{sec:CIBequilibria}, we introduce the concept of \CI-PBE and provide a sequential decomposition of the dynamic game. In Section \ref{sec:example}, we provide an example of a multiple access broadcast game that illustrates the results of Section \ref{sec:CIBequilibria}. We prove the existence of \CI-PBE for a subclass of dynamic games in Section \ref{sec:existence}. We conclude in Section \ref{sec:conclusion}.  The proofs of all of our technical results appear in the Appendices \ref{app:CIB}-\ref{app:existence}.

\subsection*{Notation}
Random variables are denoted by upper case letters, their realization by the corresponding lower case letter.
In general, subscripts are used as time index while superscripts are used to index agents.
For time indices $t_1\leq t_2$, $X_{t_1:t_2}$ (resp. $f_{t_1:t_2}(\cdot)$) is the short hand notation for the variables $(X_{t_1},X_{t_1+1},...,X_{t_2})$ (resp.  functions $(f_{t_1}(\cdot),\dots,f_{t_2}(\cdot))$).
When we consider the variables (resp. functions) for all time, we drop the subscript and use $X$ to denote $X_{1:T}$ (resp. $f(\cdot)$ to denote $f_{1:T}(\cdot)$).
For variables $X^1_t,\dots,X^N_t$ (resp. functions $f^1_t(\cdot),\dots,f^N_t(\cdot)$), we use $X_t:=(X^1_t,\dots,X^N_t)$ (resp. $f_t(\cdot):=(f^1_t(\cdot),\dots,f^N_t(\cdot))$) to denote the vector of the set of variables (resp. functions) at $t$, and $X^{-n}_t:=(X^1_t,\dots,X^{n-1}_t,X^{n+1}_t,\dots,X^N_t)$ (resp. $f^{-n}_t(\cdot):=(f^1_t(\cdot),\dots, f^{n-1}_t(\cdot),f^{n+1}_t(\cdot),\dots,f^N_t(\cdot))$) to denote all the variables (resp. functions) at $t$ except that of the agent indexed by $n$.
%For functions $f^1_t(\cdot),f^2_t(\cdot),\dots,f^N_t(\cdot)$, we use $f_t(\cdot):=(f^1_t(\cdot),f^2_t(\cdot),\dots,f^N_t(\cdot))$  to denote the vector-valued function and $f^{-n}_t(\cdot):=(f^1_t(\cdot),\dots, f^{n-1}_t(\cdot),f^{n+1}_t(\cdot),\dots,f^N_t(\cdot))$ to denote the functions except that of a user indexed  by $n$.
$\mathbb{P}(\cdot)$ and $\mathbb{E}(\cdot)$ denote the probability and expectation of an event and a random variable, respectively.
For a set $\mathcal{X}$, $\Delta(\mathcal{X})$ denotes the set of all beliefs/distributions on $\mathcal{X}$.
For random variables $X,Y$ with realizations $x,y$, $\mathbb{P}(x|y) := \mathbb{P}(X=x|Y=y)$ and $\mathbb{E}(X|y) := \mathbb{E}(X|Y=y)$.
For a strategy $g$ and a belief (probability distribution) $\pi$, we use $\mathbb{P}^g_{\pi}(\cdot)$ (resp. $\mathbb{E}^g_{\pi}(\cdot)$) to indicate that the probability (resp. expectation) depends on the choice of $g$ and $\pi$. We use $\mathbf{1}_{\{x\}}(y)$ to denote the indicator that $X=x$ is in the event $\{Y=y\}$.

\section{System Model}
\label{sec:model}
%\section{System Model}
%\label{sec:model}
Consider a dynamic game among $N$ strategic agents, indexed by $\mathcal{N}:=\{1,2,\dots,N\}$, in a system
over time horizon $\mathcal{T}:=\{1,2,\cdots,T\}$. 
Each agent $n\in\mathcal{N}$ is affiliated with a subsystem $n$. %An agent can be a player in the networked game, or a controller that controls part of the system.
At every time $t\in \mathcal{T}$, the state of the system $(C_t,X_t)$ has two components: $C_t\in \mathcal{C}_t$ denotes the public state, and $X_t := (X^1_t,X^2_t,\dots,X^N_t)\in \mathcal{X}_t := \mathcal{X}^1_t\times \mathcal{X}^2_t\times\dots\times \mathcal{X}^N_t$, where $X^n_t$ denotes the local state of subsystem $n,n\in \mathcal{N}$. The public state $C_t$ is commonly observed by every agent, and the local state $X_t^n$ is privately observed by agent $n,n\in \mathcal{N}$. 

At time $t$, each agent $n$ simultaneously selects an action $A^n_t \in \mathcal{A}^n_t$. 
%$A^n_t \in \mathcal{A}^n_t(X^n_t)\subseteq \mathcal{A}^n_t$. 
%The set of possible actions $\mathcal{A}^n_t(X^n_t)$ may depend on the local state of agent $n$ at $t$.
Given the control actions $A_t := (A^1_t,A^2_t,\dots,A^N_t)$, the public state and local states evolve as
% dynamics is given by
\begin{align}
&C_{t+1}=f_t^c(C_t,A_t,W_t^{C}),
\label{eq:Cdynamics}\\
&X_{t+1}^n=f_t^n(X_t^n,A_{t},W_t^{n,X}),\;n\in\mathcal{N},
\label{eq:Sdynamics}
\end{align}
where random variables $W_t^{C}$ and $W_t^{n,X}$ capture the randomness in the evolution of the system, and $C_1, X^1_1,X^2_1,\dots,X^N_1$ are primitive random variables.

At the end of time $t$, after the actions are taken, each agent $n\in \mathcal{N}$ observes $Y_t:=(Y^1_t,Y^2_t,\dots,Y^N_t)$, where
\begin{align}
Y^n_t = h^n_t(X^n_t,A_t,W_t^{n,Y}) \in \mathcal{Y}^n_t,
\label{eq:Ydynamics}
\end{align}
and $W_t^{n,Y}$ denotes the observation noise.
From the system dynamics \eqref{eq:Sdynamics} and the observations model \eqref{eq:Ydynamics}, we define, for any $n\in\mathcal{N},t\in\mathcal{T}$, the probabilities $p^n_t(x^n_{t+1};x^n_t,a_t)$ and $q^n_t(y^n_t;x^n_t,a_t)$ such that for all $x^n_{t+1},x^n_t\in\mathcal{X}^n_t$, $y^n_t\in\mathcal{Y}^n_t$ and $a_t\in\mathcal{A}_t:=\mathcal{A}^1_t\times\dots\times\mathcal{A}^N_t$
\begin{align}
& p^n_t(x^n_{t+1};x^n_t,a_t):= \mathbb{P}(f_t^n(x_t^n,a_{t},W_t^{n,X})=x^n_{t+1}),
\\
& q^n_t(y^n_t;x^n_t,a_t):= \mathbb{P}(h^n_t(x^n_t,a_t,W_t^{n,Y})=y^n_{t}).
\end{align}

We assume that $\mathcal{C}_t,\mathcal{X}^n_t,\mathcal{A}^n_t$ and $\mathcal{Y}^n_t$ are finite sets for all $n\in \mathcal{N}$, for all $t\in \mathcal{T}$.\footnote{The results developed in Section \ref{sec:model}-\ref{sec:CIBequilibria} for finite $\mathcal{C}_t,\mathcal{X}^n_t,\mathcal{A}^n_t$ and $\mathcal{Y}^n_t$ still hold
% when $\mathcal{C}_t,\mathcal{X}^n_t,\mathcal{A}^n_t$ and $\mathcal{Y}^n_t$ 
when they
are continuous sets under some technical assumptions. The results of Section \ref{sec:existence} require $\mathcal{A}^n_t$ to be finite for all $n\in\mathcal{N}$, $t\in\mathcal{T}$.}
We also assume that the primitive random variables $\{C_1,X^n_1, W_t^{C}, W_t^{n,X}, W_t^{n,Y}, t\in \mathcal{T}, n\in\mathcal{N}\}$ are mutually independent.

%\begin{assumption}
%\label{assum:indep}
%The primitive random variable $\{C_1,X^n_1, W_t^{C}, W_t^{n,X}, W_t^{n,O}, t\in \mathcal{T}, n\in\mathcal{N}\}$ are mutually independent.
%\end{assumption}

The actions $A_t$ and the observations $Y_t:=(Y^1_t,Y^2_t,\dots,Y^N_t)$
are commonly observed by every agent.
Therefore, at time $t$, all agents have access to the common history $H^c_t$ defined to be
\begin{align}
H^c_t := \{C_{1:t}, A_{1:t-1}, Y_{1:t-1}\}.
\label{def:Hc}
\end{align}

Including private information, the history $H_t^n$ of agent $n$'s information, $n \in \mathcal{N},$ at $t$ is given by
\begin{align}
H^n_t := \{X^n_{1:t},H^c_t\}=\{X^n_{1:t},C_{1:t}, A_{1:t-1}, Y_{1:t-1}\}.
\label{def:HO}
\end{align}

Let $\mathcal{H}^c_t$ denote the set of all possible common histories at time $t\in \mathcal{T}$, and $\mathcal{H}^n_t$ denote the set of all possible information histories for agent $n\in\mathcal{N}$ at time $t\in\mathcal{T}$.
% and $\mathcal{H}^n:=\cup_{t\in\mathcal{T}}\mathcal{H}^n_t$.

Define $H_t := \cup_{n\in\mathcal{N}} H^n_t = \{X_{1:t},H^c_t\}$ to be the history of states and observations of the whole system up to time $t$. The hirtory $H_t$ captures the system evolution up to time $t$. 
%Let $\mathcal{H}_t$ denote the set of all possible histories of states and observations at time $t$.
% and $\mathcal{H}:=\cup_{t\in\mathcal{T}}\mathcal{H}_t$.

A behavioral strategy of agent $n,n \in \mathcal{N}$, is defined as a map $g^n_t: \mathcal{H}^n_t \mapsto \Delta(\mathcal{A}^n_t)$ where 
\begin{align}
\mathbb{P}^{g^n_t}(A^n_t = a^n_t|h^n_t) := g^n_t(h^n_t)(a^n_t) \text{ for all } a^n_t \in \mathcal{A}^n_t.
\end{align}

%\red{
%If the set of an agent's possible actions depend on the current local state of an agent, we use $\mathcal{A}^n_t(X^n_t)\subseteq\mathcal{A}^n_t$ to denote the set of possible actions for agent $n$ at $t$. Then any behavioral strategy should satisfy $g^n_t(h^n_t)\in\Delta(\mathcal{A}^n_t(x^n_t))$ for $h^n_t = (x^n_{1:t},h^c_t)\in\mathcal{H}^n_t$.
%}

Let $\mathcal{G}^n_t$ denote the set of all possible behavioral strategies
\footnote{The results developed in this paper also holds when agent $n$'s set of admissible actions depends on his current private state. That is, $A^n_t \in \mathcal{A}^n_t(x^n_t) \subseteq \mathcal{A}^n_t$ and $g^n_t(h^n_t) \in \Delta(\mathcal{A}^n_t(x^n_t))$ when $h^n_t = (x^n_{1:t},h^c_t)$.} 
$g^n_t$ of user $n\in\mathcal{N}$ at time $t\in\mathcal{T}$.
%Let $\mathcal{G}^n = \bigodot_{t\in \mathcal{T}}\mathcal{G}^n_t $ and $\mathcal{G}^n = \bigodot_{n\in \mathcal{N}}\mathcal{G}^n $.

At each time $t\in \mathcal{T}$, agent $n,n \in \mathcal{N},$ has a utility 
\begin{align}
U^n_t = \phi^n_t(C_t,X_t,A_t)
\end{align}
that depends on the state of the system at $t$, including the public state and all local states, and the actions taken at $t$ by all agents. 

Let  $g = (g^1,g^2,\dots,g^N)$ denote the strategy profile of all agents, where $g^n =(g^n_1,g^n_2,\dots,g^n_T)$. 
Then, the total expected utility of agent $n$ is given by
\begin{align}
U^n(g) = \mathbb{E}^g\left[\sum_{t=1}^T U^n_t\right] = \mathbb{E}^g\left[\sum_{t=1}^T \phi^n_t(C_t,X_t,A_t)\right].
\end{align}
Each agent wishes to maximize his total expected utility. 

The problem defined above is a stochastic dynamic game with asymmetric information. 

As discussed above, signaling may occur in games of asymmetric information. 
The game instances that can be captured by our model could belong to any of the four game environments (a)-(d) described in Section \ref{sec:intro}.

\section{Solution Concept}
\label{sec:solconcept}
%\section{Solution Concept}

For non-cooperative static games with complete information (resp. incomplete information), one can use Nash equilibrium  (resp. Bayesian Nash equilibrium) as a solution concept. A strategy profile $g\!=\!(g^1,\cdots,g^N)$ is a (Bayesian) Nash equilibrium, if there in no agent $n$ that can unilaterally deviate to another strategy $g'^n$ and get a higher expected utility. One can use (Bayesian) Nash equilibrium to analyze dynamic stochastic games.
However, the (Bayesian) Nash equilibrium solution concept ignores the dynamic nature of the system and only requires optimality with respect to any unilateral deviation from the equilibrium $g$ at the beginning of the game (time $1$). Requiring optimality only against unilateral deviation at time $1$ could lead to irrational situations such as non-credible threats \cite{fudenberg1991game,osborne1994course}. In dynamic games, a desirable equilibrium $g$ should guarantee that there is no profitable unilateral deviation for any agent at any stage of the game. That is, for any $t\in\mathcal{T}$, for any realization $h_t\in\mathcal{H}_t$ of the system evolution, the strategy $g_{t:T}=(g^1_{t:T},g^2_{t:T}\cdots,g^N_{t:T})$ must be a (Bayesian) Nash equilibrium of the continuation game that follows $h_t$. This requirement is called \textit{sequential rationality} \cite{fudenberg1991game,osborne1994course}. 

In this paper we study dynamic stochastic games of incomplete asymmetric information. At time $t$, the system evolution $H_t$ is not completely known to all agents; each agent $n\in\mathcal{N}$ only observes $H_t^n$ and has to form a belief about the complete system evolution $H_t$ up to time $t$. The belief that agent $n$ forms about $H_t$ depends in general on both $H_t^n$ and $g^{-n}_{1:t}$. Knowing the strategy of the other agents, agent $n$ can make inference about other agents' private information $X^{-n}_{1:t}$ from observing their actions. 
As pointed out in Section \ref{sec:intro}, this phenomenon is called signaling in games with asymmetric information. 
Signaling results in agents' beliefs that depend on the strategy profile $g$ (see the discussion in Section \ref{sec:intro}). 
Therefore, at an equilibrium such beliefs must be consistent with the equilibrium strategies via Bayes' rule. Moreover, the sequential rationality requirement must be satisfied with respect to the agents' beliefs. We call the collection of all agents' beliefs at all times a \textit{belief system}.
A pair of strategy profile and belief system that are mutually sequentially rational and consistent form a \textit{perfect Bayesian equilibrium} (PBE).
We use PBE as the solution concept in this paper to study the dynamic game defined in Sectin \ref{sec:model}.
% \footnote{We note that there exists a refinement of PBE called sequential equilibrium that has further constraint on (off-equilibrium) belief systems \cite{fudenberg1991game}. The dynamic stochastic game considered in this paper has a Markovian dynamics with mutually independent private state evolution, and, therefore, satisfies the conditions of \cite{battigalli1996strategic} and \cite[ch.8]{fudenberg1991game}. Consequently, the set of PBEs and the set of sequential equilibria of the game defined in this paper coincides.}

%Sequential rationality has been introduced to resolve the issues \cite{osborne1994course} in dynamic games with asymmetric information.
%\begin{itemize}
%\item Discuss more on sequential rationality
%\item Talk about perfect Bayesian equilibrium (PBE)
%\item We use PBE as our solution concept\footnote{Note that the dynamic game model defined in this paper satisfies the conditions of %\cite{battigalli1996strategic} \cite[ch. 8]{fudenberg1991game}, because the Markov processes describing the private states are mutually independent. %Therefore, the set of PBEs and the set of sequential equilibria of the game defined in this paper are the same.}
%\end{itemize}
We note that the system model we use in this paper is different from the standard model of extensive game form used in the game theory literatures \cite{fudenberg1991game,osborne1994course}. Specifically, the model of Section \ref{sec:model} is a state space model (that describes the stochastic dynamics of the system), while the extensive game form is based on the intrinsic model \cite{witsenhausen1975intrinsic} whose components are nature's moves and users' actions. The two models are equivalent within the context of sequential dynamic teams \cite{witsenhausen1988equivalent}. In order to analyze the dynamic game of the state space model of Section \ref{sec:model}, we need to provide the formal definition of PBE for our model in the following.

%In this paper, we use sequential rationality to define our solution concept for the game defined in Section \ref{sec:model}. We will compare our model and solution concept with the standard extensive game form which has been used to model and analyze dynamic games with asymmetric information.

\subsection{Perfect Bayesian Equilibrium}
To provide a formal definition of PBE for our state space model defined in Section \ref{sec:model}, we first define histories of states, beliefs and signaling-free beliefs on histories of states. 

\begin{definition}[History of States]
The history of states at each time $t$ is defined to be $X_{1:t}$.
\end{definition}

Note that the history of states contains the trajectory of local state $X^n_{1:t}$ that is private information of agent $n, n\in\mathcal{N}$.
%is different from the history of observations defined by \eqref{def:HO}.
%The history of states $X_{1:t}$ includes the private state processes of all users while the history of observations $H^n_t$ contains only observations available to user $n$ at time $t$.

\begin{definition}[Belief System]
Let $\mu^n_t: \mathcal{H}^n_t \mapsto \Delta(\mathcal{X}_{1:t})$. For every history $h^{n}_t\in\mathcal{H}^n_t$, the map $\mu^n_t$ defines a belief for agent $n\in\mathcal{N}$ at time $t\in\mathcal{T}$ which is a probability distribution on the histories of states $X_{1:t}$. %be a map that defines a belief 
%for every observation history $h^{n}_t\in\mathcal{H}^n_t$ for agent $n\in\mathcal{N}$ at time $t\in\mathcal{T}$
%defines a probability distribution on histories of states $X_{1:t}$.
The collection of maps $\mu:=\{\mu^n_t,n\in\mathcal{N}, t\in\mathcal{T}\}$ is called a belief system on histories of states. 

\end{definition}

That is, given a belief system $\mu$, agent $n\in\mathcal{N}$ assigns the probability distribution $\mu^n_t(h^{n}_t)$ on $\mathcal{X}_{1:t}$ conditioning on the realized history of observations $h^n_t\in\mathcal{H}^n_t$ at $t\in\mathcal{T}$, by
% such that user $n$'s (conditional) belief on $\{X_{1:t} = x_{1:t}\}$ equals to $\mu^n_t(h^{n}_t)(x_{1:t})$. 
%That is, for every observation history $h^{n}_t\in\mathcal{H}^n_t$, $\mu^n_t(h^{n}_t)$ defines a probability distribution on $X_{1:t}$ such that user $n$'s (conditional) belief on $\{X_{1:t} = x_{1:t}\}$ equals to $\mu^n_t(h^{n}_t)(x_{1:t})$. 
%The collection of maps $\mu:=\{\mu^n_t,n\in\mathcal{N},t\in\mathcal{T}\}$ is called a belief system on histories of states. 
%Under a belief system $\mu$, the (conditional) belief on $\{X_{1:t} = x_{1:t}\}$ given $h^n_t\in\mathcal{H}^n_t$ for any $n\in\mathcal{N}$, $t\in\mathcal{T}$, is denoted by
\begin{align}
\mathbb{P}_{\mu}( x_{1:t}|h^n_t):=\mu^n_t(h^{n}_t)(x_{1:t}).
\label{eq:measuremu}
\end{align}

Then, given the beliefs $\mu^n_t(h^{n}_{t})$ for agent $n\in\mathcal{N}$ at $h^{n}_{t} = (x^n_{1:t},h^{c}_{t})\in\mathcal{H}^{n}_{t}$ and a strategy $g_t$ at $t\in\mathcal{T}$, 
when agent $n$ takes an action $a^n_t\in\mathcal{A}^n_t$, his belief about the system following $(h^n_t,a^n_t)$ is given by $\mathbb{P}^{g_t}_{\mu}(x_{1:t+1},y_t,a_t|h^n_t,a^n_t)$ for any $x_{1:t+1}\in\mathcal{X}_{1:t+1}, y_t\in\mathcal{Y}_t,a_t \in\mathcal{A}_t$, where
\begin{align}
 &\mathbb{P}^{g_t}_{\mu}(x_{1:t+1},y_t,a_t|h^n_t,a^n_t) \nonumber\\
:= &\mu^n_t(h^{n}_{t})(x_{1:t})\prod_{k\in\mathcal{N}} p^k_t(x^k_{t+1};x^k_t,a_t)q^k_t(y^k_{t};x^k_t,a_t) \nonumber\\
   &\quad\prod_{k\neq n}g^{k}_t(x^k_{1:t},h^c_t)(a^k_t).
\label{eq:conditionalupdate}
\end{align}

\begin{definition}[Signaling-Free Beliefs]
The signaling-free belief system $\hat \mu:=\{\hat \mu^n_t: \mathcal{H}^n_t \mapsto \Delta(\mathcal{X}_{1:t}),n\in\mathcal{N},t\in\mathcal{T}\}$ is defined on histories of states such that for each $n\in\mathcal{N}$, $t\in\mathcal{T}$, and
$h^{n}_t:=(x^n_{1:t},c_{1:t},a_{1:t-1},y_{1:t-1})\in\mathcal{H}^n_t$
\begin{align}
\hat \mu^n_t(h^n_t)(x_{1:t}) := &\mathbb{P}_{(A_{1:t-1}=a_{1:t-1})}(x_{1:t}|y_{1:t-1},x^n_{1:t})\nonumber\\
 &\text{ for any }x_{1:t} \in \mathcal{X}_{1:t}.
\label{eq:sigfreebef}
\end{align}
The right hand side of \eqref{eq:sigfreebef} gives the conditional probability of $\{X_{1:t} = x_{1:t}\}$ given $\{Y_{1:t-1}=y_{1:t-1},X^n_{1:t}=x^n_{1:t}\}$ when $A_{1:t-1}=a_{1:t-1}$. This conditional probability is computed using the realization $h_t^n$ of agent $n$'s  information, the subsystem dynamics \eqref{eq:Sdynamics}, and the observation model \eqref{eq:Ydynamics}.

\end{definition}

%\begin{definition}[Signaling Free Beliefs on States]
%The signaling free belief system $\hat \mu:\mathcal{H} \mapsto \cup_{t=1}^T\Delta(\mathcal{X}^{t\times N})$ on the histories of states is defined for each $h^{n}_t:=(c_{1:t},a_{1:t},y_{1:t-1})\in\mathcal{H}$ by
%\begin{align}
%\hat \mu(h^n_t)(x_{1:t}) := \mathbb{P}^(x_{1:t}|h^n_t),
%\label{eq:sigfreebef}
%\end{align}
%
%where $\mathbb{P}^(x_{1:t}|h^n_t)$ is computed using the realization $h_t^n$ o the information available to user $n$ up to time $t$, the system dynamics (\ref{eq:Sdynamics}), and the observations (\ref{eq:Ydynamics}).
%%where $g^{a_{1:t-1}}$ is the strategy such the actions $a_{1:t-1}$ are taken with probability one at each time of the decisions. That is, 
%%\begin{align}
%%g^{a_{1:t-1}}(h^n_\tau)(a^n_\tau) = 1
%%\label{eq:gfixeda}
%%\end{align}
%%for any user $n$ at any time $\tau\leq t-1$ with every history $h^n_\tau$. 
%%
%\end{definition}

Note that the signaling-free belief $\hat \mu^n_t(h^n_t)$ is \textit{strategy-independent}. One can think $\hat \mu^n_t(h^n_t)$ as the belief generated by the open-loop strategy $(a_1,a_2,\cdots,a_{t-1})$, so there is no signaling and strategy-dependent inference present in the belief system. The role of signaling-free belief will become evident when we talk about consistency in the definition of PBE for the state space model described in Section \ref{sec:model}.

The beliefs defined above are used by the agents to evaluate the performance of their strategies. 
Sequential rationality requires that at any time instant, each agent's strategy is his best response under his belief about the system states.

This relation between a strategy profile $g$ and a belief system $\mu$ is formally defined as follows.

\begin{definition}[Sequential Rationality]
\label{def:seqrational}
A pair $(g,\mu)$ satisfies \textit{sequential rationality} if for every $n\in\mathcal{N}$,
$g^{n}_{t:T}$ is a solution to
\begin{align}
\sup_{g'^{n}_{t:T} \in \mathcal{G}^{n}_{t:T}}
\mathbb{E}^{g'^{n}_{t:T},g^{-n}}_{\mu}\left[\sum_{\tau=t}^T \phi^{n}_\tau(C_\tau,X_\tau,A_\tau)| h^{n}_t \right]
\label{eq:seqrational}
\end{align} 
for every $t\in\mathcal{T}$ and every history $h^{n}_t\in\mathcal{H}^{n}_t$, where $\mathbb{E}^{g'^{n}_{t:T},g^{-n}}_{\mu}[\cdot|h^n_t]$ is computed using the probability measure generated from \eqref{eq:measuremu}-\eqref{eq:conditionalupdate} using the belief system $\mu$ and the strategy profile $(g'^{n}_{t:T},g^{-n})$ given the realization $h^n_t$.
\end{definition}

The above definition of sequential rationality does not place any restriction on the belief system. However, rational agents should form their beliefs based on the strategies used by other agents. This consistency requirement is defined as follows.

\begin{definition}[Consistency]
\label{def:consistency}
A pair $(g,\mu)$ satisfies \textit{consistency} if $\mu$ can be computed by Bayes' rule whenever possible. 
That is, for $n\in \mathcal{N}, t\in\mathcal{T},$ such that $\mathbb{P}^{g_t}_{\mu}(y_t,a_t|h^n_t,a^n_t)>0$,
\begin{align}
\mu^n_{t+1}(h^{n}_{t+1})(x_{1:t+1})
= &\mathbf{1}_{\{x^n_{t+1}\}}(h^n_{t+1})
\frac{\mathbb{P}^{g_t}_{\mu}(x_{1:t+1},y_t,a_t|h^n_t,a^n_t)}{\mathbb{P}^{g_t}_{\mu}(x^n_{t+1},y_t,a_t|h^n_t,a^n_t)}
\nonumber\\
&\text{ for any }x_{1:t+1} \in \mathcal{X}_{1:t+1}
\label{eq:consist_positive}
\end{align}
where $\mathbb{P}^{g_t}_{\mu}(\cdot|h^n_t,a^n_t)$ is the probability measure given by \eqref{eq:conditionalupdate}.
Furthermore, when $\mathbb{P}^g_{\mu}(y_t,a_t|h^n_t,a^n_t)=0$, $\mu^n_{t+1}(h^{n}_{t+1})$ is a probability distribution in $\Delta(\mathcal{X}_{1:t+1})$ such that
% $\mu^n_{t+1}(h^{n}_{t+1})$ satisfies
\begin{align}
\mu^n_{t+1}(h^{n}_{t+1})(x_{1:t+1})= 0 
\text{ if } \hat\mu^n_{t+1}(h^{n}_{t+1})(x_{1:t+1}) = 0.
\label{eq:consist_sigfree}
\end{align}

\end{definition}
%\red{Hamid: Talk about the arbitrary definition of belief?}

Note that the signaling-free belief system $\hat \mu$ is used in \eqref{eq:consist_sigfree} of the definition for consistency. We will explain in the discussion below the importance of signaling-free beliefs on agents' rational behavior.

Using the above definitions, we define PBE for the stochastic dynamic game with asymmetric information described by the model of Section \ref{sec:model}.

\begin{definition}[Perfect Bayesian Equilibrium]
A pair $(g,\mu)$ is called a \textit{perfect Bayesian equilibrium} (PBE) if it satisfies \textit{sequential rationality} and \textit{consistency}.
\end{definition}

%\subsection{Comparison with the extensive game form}
\subsection{Discussion}

As we mentioned earlier, the state space model and the extensive game form are different but equivalent representations of sequential dynamic teams. 
We discuss the connection between these two models for dynamic games. The table below summarizes the key components of our state space model and the extensive game form (see \cite{fudenberg1991game, osborne1994course}).

\begin{center}

\begin{tabular}{|c|c|}
\hline
State Space Model & Extensive Game Form \\
\hline
State $X_t$ & No State \\
%History of States and Observations $H_t$ & History of Actions (including nature's moves)\\
History $H_t$ & History of Actions\\
%History of Observations $H^n_t$ & Information Sets \\
History $H^n_t$ & Information Sets \\
%Beliefs on Histories of States $\mu^n_t(h^n_t)(x_{1:t})$ & Beliefs on Nodes in an Information Set\\
Belief $\mu^n_t(h^n_t)(x_{1:t})$ & Belief on an Information Set\\
%Support of a Signaling-Free Belief $\hat\mu^n_t(h^n_t)(x_{1:t})$ & Nodes in an Information Set \\
Support of $\hat\mu^n_t(h^n_t)(x_{1:t})$ & Nodes in an Information Set \\
%Perfect Bayesian Equilibrium & Perfect Bayesian Equilibrium \\
PBE & PBE \\
\hline
\end{tabular}

\end{center}
%\caption{Key elements of state space model and extensive game form}

The state variable $X_t$ in the state space model allows us to easily describe the system dynamics by \eqref{eq:Sdynamics}. 
Without an explicit state variable in the extensive form, it may be complex to describe and analyze the system dynamics.
In the state space model, the system's evolution is captured by the history of states and observations $H_t$. This is the analogue of the history of (agents' and nature's) actions in the extensive game form that captures the game's evolution trajectory.
The history of information $H^n_t$ defined in the state space model includes all information available to agent $n$ at time $t$. 
This history determines what agent $n$ knows about the system, and is the analogue of an information set in the extensive game form.
Similarly, agent $n$'s belief on histories of states (conditional on $H^n_t$) in the state space model is the analogue of agent $n$'s belief over an information set in the extensive game form.
%Therefore, user $n$'s belief on histories of states (conditional on $H^n_t$) provides the same function as a belief on an information set. 
%However, the belief $\mu(h^n_t)(x_{1:t})$ is different from a belief on an information set since they have different domain.

Generally, a belief $\mu^n_t(h^n_t)(x_{1:t})$ can have a fixed support that includes the entire state space $\mathcal{X}_{1:t}$. However, a belief on an information set has a variable support that includes nodes in that particular information set. %(which depends on the system dynamics and the information structure of the game). 
Given a strategy profile, one can determine the belief using the Bayes' rule, given by \eqref{eq:consist_positive}, whenever possible for our state space model and (similarly) for the extensive game form model. However, when the denominator is zero in \eqref{eq:consist_positive}, or we reach an information set of measure zero in the extensive game form, one needs to assign values for the belief on $\mathcal{X}_{1:t}$ and on the nodes of the information set. In the extensive game form, the consistency condition allows for any arbitrary probability distribution over the nodes of the (reached) information set of measure zero. However, in our state space model we need to make sure that the belief assigned is consistent with the dynamics of the system. As a result, the belief does not necessarily assign a positive probability to a history of states and must be more carefully defined. This is where signaling-free beliefs play an important role.

To establish the equivalence between the belief $\mu^n_t(h^n_t)(x_{1:t})$ in our state space model and the belief on the information set of the corresponding extensive game form, we introduce the signaling-free beliefs. 
The signaling-free belief $\mu^n_t(h^n_t)(x_{1:t})$ defined by \eqref{eq:sigfreebef} for $h^n_t=(x^n_{1:t},c_{1:t},a_{1:t-1},y_{1:t-1})$ is constructed by actions $A_{1:t-1}=a_{1:t-1}$ conditioned on the history of observations $y_{1:t-1},x^n_{1:t}$ using the system dynamics. In forming a signaling-free belief no underlying strategy profile is assumed, and we do not make any further inference by tracing back how the observed actions are generated (i.e. the observed actions are generated by an open loop strategy). Therefore, if a history of states $x_{1:t}$ does not belong to the support of the signaling-free belief (i.e. $\hat \mu^n_t(h^n_t)(x_{1:t})=0$), this history of states $x_{1:t}$ can not happen under any possible strategy profile. A rational agent should not assign positive probability on any history of states that is outside the support of the signaling-free belief. This leads to the second part of the consistency requirement \eqref{eq:consist_sigfree}. % The support of a signaling free belief is actually equivalent to the nodes in an information set in the extensive game form.
With this additional requirement, the definition of consistency in our state space model is the analogue of the consistency in the extensive game form, and the definitions of PBE in the two models become identical.
%\footnote{If the entries of the matrix transition probability are all positive for all $t$, then the support of the signaling free beliefs become the entire state space and condition \eqref{eq:consist_sigfree} is trivially satisfied,  since the realization of any state given an arbitrary history has a positive probability.}

We note that the signaling-free beliefs are strategy-independent. In systems where any agent's belief on system's states is strategy-independent (e.g. the finite games considered in \cite{nayyar2014Common} and linear-Gaussian systems \cite{gupta2014common}), one can show that for any strategy profile $g$, the only consistent belief system 
is the signaling-free belief system $\hat \mu$. In this type of systems, consistency is trivially satisfied using the signaling-free belief system $\hat \mu$. As a result, it is sufficient to verify sequentially rationality to establish a PBE for systems with strategy-independent beliefs.
% (see \cite{nayyar2014Common,gupta2014common}).

\section{Common Information Based Perfect Bayesian Equilibria and Sequential Decomposition}
\label{sec:CIBequilibria}
In this section, we introduce the common information based (\CI) belief system and \CI\ beliefs. The \CI\ beliefs generally depend on the agents' strategies because of the presence of signaling in dynamic games with asymmetric information. 
We use \CI\ beliefs to construct \CI\ strategy profiles for the agents. 
Using the concept of \CI\ belief system and \CI\ strategy profile, we define a subclass of PBE called \textit{common information based perfect Bayesian equilibria} (\CI-PBE).
The main result of this section provides a sequential decomposition for the dynamic game model in Section \ref{sec:model}; this decomposition leads to a backward induction algorithm to compute \CI-PBE.

\subsection{Preliminaries}

Based on common histories, we first define \CI\ signaling-free belief system 
\begin{definition}[\CI\ Signaling-Free Belief System]
The \CI\ signaling-free belief system is $\hat \gamma :=\{\hat \gamma_t:\mathcal{H}^c_t \mapsto \Delta(\mathcal{X}_t),t\in\mathcal{T}\}$ 
%such that 
%$\hat \gamma_t(h^c_t): \mathcal{H}^c_t \mapsto \Delta(\mathcal{X}_t)$ at time $t,t\in\mathcal{T},$ on states $\mathcal{X}_t$ 
where for each $t\in\mathcal{T}$ and
$h^c_t = (c_{1:t},a_{1:t-1},y_{1:t-1})\in \mathcal{H}^c_t $, $\hat \gamma_t(h^c_t)$ is a belief on states $X_t$, with
\begin{align}
\hat \gamma_t(h^c_t)(x_t) \!:=\! \mathbb{P}_{\{A_{1:t-1}=a_{1:t-1}\}}(x_t|y_{1:t-1}) \text{ for } x_t \in \mathcal{X}_t.
\label{eq:sigfreebeft}
\end{align}

\end{definition}

The right hand side of \eqref{eq:sigfreebeft} is interpreted in the same way as the right hand side of \eqref{eq:sigfreebef}.
%We call $\hat \gamma_t :=\{\hat \gamma_t,t\in\mathcal{T}\}$ a \CI\ signaling-free belief system.
Note that, $\hat \gamma_t(h^c_t)(x^{-n}_t) = \hat \mu^n_t(h^n_t)(x^{-n}_t)$ from its definition, when $h^n_t = (x^n_{1:t},h^c_t)$ for any $n\in\mathcal{N}$.
We use 
\begin{align}
\hat \Pi_t := \hat\gamma_t(H^c_t)
\end{align}
to denote the \CI\ signaling-free belief at time $t$. Similar to signaling-free beliefs on histories of states defined in \eqref{eq:sigfreebef}, 
the \CI\ signaling-free belief $\hat \Pi_t$ depends only on the system dynamics and the observation model. 

The \CI\ signaling-free beliefs have the following dynamics.
\begin{lemma}[Evolution of \CI\ Signaling-Free Beliefs]
\label{lm:SFupdate}
The \CI\ signaling-free beliefs $\{\hat \Pi_t,t\in\mathcal{T}\}$ can be updated by
\begin{align}
\hat \Pi_{t+1} = \prod_{n=1}^N \hat \Pi^n_{t+1} \text{, where}
\label{eq:SFindep}
%\hat\psi_t(\hat \Pi_{t},A_t,Y_t),
\end{align}
\begin{align}
&\hat \Pi^n_{t+1}=\hat\psi^n_t(Y^n_t,A_t,\hat \Pi^n_{t}),
\\
&\hat\psi^n_t(y^n_t,a_t,\hat \pi^n_{t})(x^n_{t+1}) \nonumber\\
:=
&\frac{\sum_{x^n_{t} \in \mathcal{X}^n_t}p^n_t(x^n_{t+1};x^n_t,a_t) q^n_t(y^n_{t};x^n_t,a_t) \hat\pi^{n}_{t}(x^n_{t})}
	{\sum_{x'^n_{t} \in \mathcal{X}^n_t}q^n_t(y^n_{t};x'^n_t,a_t)\hat\pi^{n}_{t}(x'^n_{t})} .
\end{align}
\end{lemma}
\begin{proof}
See Appendix \ref{app:CIB}.
\end{proof}

%From the definition, the \CI\ signaling-free belief $\hat \Pi_t$ is available to every agent at time $t$.
%Therefore, all agents can use \CI\ signaling-free beliefs to evaluate the performance of their strategies. 
%However, each agent should take into account other agents' signaling in their actions, and construct a belief that is more accurate than the \CI\ signaling-free belief. For that matter, we define general \CI\ belief systems based on common histories of the agents.

Similar to the belief system defined in Section \ref{sec:model}, we need a belief system to form an equilibrium.
We define \CI\ belief systems based on the agents' common histories together with \CI\ update rules. 
\begin{definition}[\CI\ Belief System and \CI\ Update Rule]
A collection of maps $\gamma :=\{\gamma_t:\mathcal{H}^c_t \mapsto \Delta(\mathcal{X}_t), t\in \mathcal{T}\}$ is called a \CI\ belief system.
% where $\gamma_t: \mathcal{H}^c_t \mapsto \Delta(\mathcal{X}_t)$ for every $t\in\mathcal{T}$. 
A set of belief update functions $\psi = \{\psi_t^{n}: \mathcal{Y}^n_t \times \mathcal{A}_t\times \mathcal{C}_t\times\Delta(\mathcal{X}_t) \times\Delta(\mathcal{X}_t)   \mapsto \Delta(\mathcal{X}^n_t), n\in\mathcal{N}, t\in\mathcal{T}  \}$
%where $\psi_t^{n}: \Delta(\mathcal{X}_t) \times\Delta(\mathcal{X}_t) \times \mathcal{A}^n_t\times \mathcal{Y}^n_t  \mapsto \Delta(\mathcal{X}^n_t)$,  
is called a \CI\ update rule.

\end{definition}

From any \CI\ update rule $\psi$, we can construct a \CI\ belief system $\gamma_{\psi}$ by the following inductive construction:
\begin{enumerate}
\item $\gamma_{\psi,1}(h^c_{1})(x_1) :=  \mathbb{P}(x_1) = \prod_{n\in\mathcal{N}} \mathbb{P}(x^n_1) \quad \forall x_1\in \mathcal{X}_1$.
\item At time $t+1$, after $\gamma_{\psi,t}(h^c_t)$  is defined, set
\begin{align}
&\gamma_{\psi,t+1}(h^c_{t+1})(x^n_{t+1}) \nonumber\\
:= &\psi^{n}_{t}(y^n_t,a_t,c_t,\gamma_{\psi,t}(h^c_{t}),\hat \gamma_t(h^c_t))(x^n_{t+1}),
\label{eq:psiB}
\end{align}
\begin{align}
\gamma_{\psi,t+1}(h^c_{t+1})(x_{t+1}) := &\prod_{n=1}^N \gamma_{\psi,t+1}(h^c_{t+1})(x^n_{t+1}),
\end{align}
for every history $h^{c}_{t+1} = (h^{c}_{t},c_{t+1},a_t,y_t)\in\mathcal{H}^{c}_{t+1}$ and for all $x_{t+1}\in \mathcal{X}_{t+1}$.
\end{enumerate}
For a \CI\ belief system $\gamma_{\psi}$, we use $\Pi^{\gamma_\psi}_t$ to denote the belief, under $\gamma_\psi$, on $X_t$ conditional on $H^c_t$; that is,
\begin{align}
&\Pi^{\gamma_\psi}_t:= \gamma_{\psi,t}(H^c_t ) \in \Delta(\mathcal{X}_t).
\end{align}
We also define the marginal beliefs on $X^n_t$ at time $t$ as
\begin{align}
&\Pi^{n,\gamma_\psi}_t(x^n_t) := \gamma_{\psi,t}(H^c_t )(x^n_t) \quad \forall x^n_t \in \mathcal{X}^n_t.
\end{align}

Since the \CI\ beliefs $\{\Pi^{\gamma_\psi}_t,t\in\mathcal{T}\}$ are common information to all agents, all agents can use $\Pi^{\gamma_\psi}_t$ to evaluate the performance of their strategies. Furthermore, if a \CI\ update rule $\psi$ is properly chosen, the \CI\ signaling-free belief $\hat\Pi_t$ and the \CI\ belief $\Pi^{\gamma_\psi}_t$ together can summarize the agents' common knowledge about the current system states $X_t$ from all previous actions $A_{1:t-1}$ and observations $Y_{1:t-1}$ available to all of them at time $t$. 
%Then the results from stochastic control \cite{kumar1986stochastic} suggest that $\hat\Pi_t$ and $\Pi^{\gamma_\psi}_t$ can serve as the sufficient statistics for the agents to make decisions. 
This motivate the concept of \CI\ strategies defined below.

%Given a \CI\ belief system $\gamma_{\psi}$, we consider \CI\ strategies where
%each user $n$ makes his decision at time $t$ based on
%$X^n_t$ and $\Pi^{\gamma_{\psi}}_t$.
\begin{definition}[\CI\ Strategy Profile]
We call a set of functions $\lambda = \{\lambda_t^{n}: \mathcal{X}^n_t \times\mathcal{C}_t\times\Delta(\mathcal{X}_t)\times\Delta(\mathcal{X}_t) \mapsto \Delta(\mathcal{A}^n_t), n\in\mathcal{N}, t\in\mathcal{T}  \}$ 
%where $\lambda_t^{n}: \mathcal{X}^n_t \times\mathcal{C}\times\Delta(\mathcal{X}_t)\times\Delta(\mathcal{X}_t) \mapsto \Delta(\mathcal{A}^n_t)$,
a \CI\ strategy profile. 
\end{definition}

For notational simplicity, let $\mathcal{B}_t:= \mathcal{C}_t\times \Delta(\mathcal{X}_t)\times\Delta(\mathcal{X}_t)$ and 
\begin{align}
b_t=(c_t, \pi_t,\hat\pi_t) \in \mathcal{B}_t
\end{align}
denote the realization of the part of common information used in a \CI\ strategy.

If agent $n$ uses a \CI\ strategy $\lambda_t^{n}$, then any action $a^n_t\in\mathcal{A}^n_t$ is taken by agent $n$ at time $t$ with probability $\lambda_t^{n}(x^n_t,b_t)(a^n_t)$ when $X^n_t = x^n_t \in \mathcal{X}^n_t$ $(C_t, \Pi^{\gamma_{\psi}}_t,\hat\Pi^{\gamma_{\psi}}_t) = b_t \in \mathcal{B}_t$.
Note that the domain $\mathcal{X}^n_t\times\mathcal{B}_t$ of a \CI\ strategy $\lambda^n_t$ is different from the domain $\mathcal{H}^n_t$ of a behavioral strategy $g^n_t$. However, given a \CI\ strategy profile $\lambda$ and a \CI\ update rule $\psi$, we can construct a behavioral strategy profile $g\in\mathcal{G}$ by
\begin{align}
&g_t^{n}(h^{n}_t) := \lambda_t^{n}(x^n_t,c_t,\gamma_{\psi,t}(h^c_t),\hat \gamma_t(h^c_t)).
\label{eq:getaS}
\end{align}

In the following we provide a definition of a \CI\ belief system consistent with a \CI\ strategy profile.

%In analogy with PBE, we define consistent \CI\ beliefs for \CI\ strategies.

\begin{definition}[Consistency]
\label{def:CIB_consistency}
For a given \CI\ strategy $\lambda^n_t$ of user $n\in\mathcal{N}$ at $t\in\mathcal{T}$, we call a belief update function $\psi^n_t$ consistent with $\lambda^n_t$ if \eqref{eq:tilpsiS} below is satisfied when the denominator of \eqref{eq:tilpsiS} is non-zero;
\begin{align}
&\psi^{n}_t(y^n_t,a_t,b_t)(x^n_{t+1}) \nonumber\\
=&\frac{\sum_{x^n_{t} \in \mathcal{X}^n_t}p^n_t(x^n_{t+1};x^n_t,a_t) \eta^n_t(x^n_{t},y^n_t,a_t,b_t)\pi^n_t(x^n_t)}
	{\sum_{x'^n_{t} \in \mathcal{X}^n_t}\eta^n_t(x'^n_{t},y^n_t,a_t,b_t)\pi^n_t(x'^n_t)} ,
\label{eq:tilpsiS} 
%\\
%&\quad \text{ when the denominator of \eqref{eq:tilpsiS} is non-zero;} \nonumber
\end{align}
where
\begin{align}
\eta^n_t(x^n_{t},y^n_t,a_t,b_t)\!:=\! q^n_t(y^n_{t};x^n_t,a_t)\lambda^{n}_t(x^n_{t},b_t)(a^n_t).
\end{align}
When the denominator of \eqref{eq:tilpsiS} is zero, 
\begin{align}
\psi^{n}_t(b_t,a_t,y^n_t)(x^n_{t+1}) = 0 
\text{ if } \hat\psi^{n}_t(\hat\pi_{t},a_t,y^n_t)(x^n_{t+1}) = 0.
\label{eq:tilpsiSD0} 
\end{align}
For any $t\in\mathcal{T}$, if $\psi^n_t$ is consistent with $\lambda^n_t$ for all $n\in\mathcal{N}$, we call $\psi_t$ consistent with $\lambda_t$.
If $\psi_t$ is consistent with $\lambda_t$ for all $t\in\mathcal{T}$, we call the \CI\ update rule $\psi=(\psi_1,\dots,\psi_T)$ consistent with the \CI\ strategy profile $\lambda=(\lambda_1,\dots,\lambda_T)$.

\end{definition}

\begin{remark}
Note that when the denominator of \eqref{eq:tilpsiS} is zero, $\psi^{n}_t(b_t,a_t,y^n_t)$ can be arbitrarily defined as a probability distribution in $\Delta(\mathcal{X}_{t+1})$ satisfying \eqref{eq:tilpsiSD0}  and consistency still holds.
%and $\hat \pi_{t+1}(x^n_{t+1}) \neq 0$, $\psi^{n}_t(\cdot)(x^n_{t+1})$ can be arbitrarily defined and consistency is still satisfied.
One simple choice is to set $\psi^{n}_t(y^n_t,a_t,b_t) =\hat\psi^{n}_t(y^n_t,a_t,\hat\pi_{t})$ when the denominator of \eqref{eq:tilpsiS} is zero; this choice trivially satisfies \eqref{eq:tilpsiSD0}. 
Thus, for any \CI\ strategy profile $\lambda$, there always exists at least a \CI\ update rule that is consistent with $\lambda$.

%One possible value of $\psi^{n}_t(\cdot)(x^n_{t+1})$ in this situation is the \CI\ signaling-free belief $\hat\pi_t(x^n_{t+1})$ 
% such that
%$\psi^{\lambda,n}_t(\pi_{t},\hat\pi_t,a_t,y^n_t)(x^n_{t+1}) = \hat\pi_t(x^n_{t+1})$ when the denominator of \eqref{eq:tilpsiS} is zero.

%For any \CI\ strategy profile $\lambda$, we can define a \CI\ update rule $\psi^{\lambda}$ such that $\psi^{\lambda}$ is given by \eqref{eq:tilpsiS} when the denominator is non-zero, and $\psi^{\lambda,n}_t(\pi_{t},\hat\pi_t,a_t,y^n_t)(x^n_{t+1}) = \hat\pi_t(x^n_{t+1})$ when the denominator of \eqref{eq:tilpsiS} is zero. This \CI\ update rule $\psi^{\lambda}$ is always consistent with $\lambda$.
%It means that for any \CI\ strategy profile $\lambda$, there always exists at least a \CI\ update rule that is consistent with $\lambda$.
\end{remark}

The following lemma establishes the relation between the consistency conditions given by Definition \ref{def:consistency} and  \ref{def:CIB_consistency}.
%connects a \CI\ strategy profile along with its consistent \CI\ belief 
%to a strategy profile along with its consistent belief for the dynamic game.
\begin{lemma}
%(Construction of Strategy Profile and Belief System from \CI\ Strategy Profile and \CI\ Belief System)
\label{lm:befupdate}

If  $\lambda$ is a \CI\ strategy profile along with its consistent \CI\ update rule $\psi$,
there exists a pair, denoted by $(g,\mu)=f(\lambda,\psi)$, such that $g$ is the strategy profile constructed by \eqref{eq:getaS} from $(\lambda,\psi)$, and $\mu$ is a belief system consistent with the strategy profile $g$. 
Furthermore, for all $h^{n}_t\in\mathcal{H}^n_t$, $x_{1:t} \in \mathcal{X}_{1:t}$
\begin{align}
&\mu^n_{t}(h^n_{t})(x_{1:t}) = 
\mathbf{1}_{\{x^n_{1:t}\}}(h^n_{t}) \prod_{k\neq n}\mu^c_{t}(h^c_{t})(x^k_{1:t}) 
%\text{ for all } x_{1:t} \in \mathcal{X}_{1:t},
\label{eq:mut}
\end{align}
where $\mu^c_{t}: \mathcal{H}^c_t \mapsto \Delta(\mathcal{X}_{1:t})$ 
satisfies the relation 
%is a map from any common history to a belief on $X_{1:t}$ 
%such that for all $h^{c}_t\in\mathcal{H}^c_t$, $k\in \mathcal{N}$ and $x^k_t \in \mathcal{X}^k_t$ 
\begin{align}
& \mu^c_{t}(h^c_t)(x^k_t) = \sum_{\mathclap{x^{-k}_{1:t}\in \mathcal{X}^{-k}_{1:t}, x^{k}_{1:t-1}\in \mathcal{X}^{k}_{1:t-1} }}
\mu^c_{t}(h^c_t)(x_{1:t})
= \gamma_{\psi,t}(h^c_t)(x^{k}_t)
\label{eq:muct}
\end{align}
for all $h^{c}_t\in\mathcal{H}^c_t$, $k\in \mathcal{N}$ and $x^k_t \in \mathcal{X}^k_t$.

%
%$\mu$ is a belief system consistent with the strategy profile $g\in \mathcal{G}$, where
%\begin{align}
%&g_t^{n}(h^{n}_t) := \lambda_t^{n}(x^n_t,c_t,\gamma_{\psi,t}(h^c_t),\hat \gamma_t(h^c_t))
%%= \lambda_t^{n}(x^n_t,c_t,\pi^{\gamma_\psi}_t,\hat \pi_t)
%\label{eq:getaS}
%\end{align}
%for all $n\in\mathcal{N}$, $t\in\mathcal{T}$, and $\mu = \{\mu^n_t,n\in\mathcal{N}, t\in\mathcal{T}\}$ is a belief system such that, for all $h^{n}_t\in\mathcal{H}^n_t$, $x_{1:t} \in \mathcal{X}_{1:t}$
%\begin{align}
%&\mu^n_{t}(h^n_{t})(x_{1:t}) = 
%\mathbf{1}_{\{x^n_{1:t}\}}(h^n_{t}) \prod_{k\neq n}\mu^c_{t}(h^c_{t})(x^k_{1:t}) 
%%\text{ for all } x_{1:t} \in \mathcal{X}_{1:t},
%\label{eq:mut}
%\end{align}
%where $\mu^c_{t}: \mathcal{H}^c_t \mapsto \Delta(\mathcal{X}_{1:t})$ 
%satisfies the relation 
%%is a map from any common history to a belief on $X_{1:t}$ 
%%such that for all $h^{c}_t\in\mathcal{H}^c_t$, $k\in \mathcal{N}$ and $x^k_t \in \mathcal{X}^k_t$ 
%\begin{align}
%& \mu^c_{t}(h^c_t)(x^k_t) = \sum_{\mathclap{x^{-k}_{1:t}\in \mathcal{X}^{-k}_{1:t}, x^{k}_{1:t-1}\in \mathcal{X}^{k}_{1:t-1} }}
%\mu^c_{t}(h^c_t)(x_{1:t})
%= \gamma_{\psi,t}(h^c_t)(x^{k}_t)
%\label{eq:muct}
%\end{align}
%for all $h^{c}_t\in\mathcal{H}^c_t$, $k\in \mathcal{N}$ and $x^k_t \in \mathcal{X}^k_t$.
%%Furthermore, $\mu$ is consistent with $g$.

\end{lemma}
\begin{proof}
See Appendix \ref{app:CIB}.
\end{proof}

Lemma \ref{lm:befupdate} implies that using a \CI\ strategy profile $\lambda$ along with its consistent update rule $\psi$ we can construct a behavioral strategy profile $g$ along with its consistent belief system $\mu$.

Note that, equation \eqref{eq:mut} in Lemma \ref{lm:befupdate} implies that for any agent $n$, his local states $X^n_{1:t}$ are independent of $X^{-n}_{1:t}$ under $\mu$ conditional on any history $h^{n}_t\in\mathcal{H}^{n}_t$. 
Furthermore, the conditional independence described by \eqref{eq:mut} still holds even when agent $n$ uses another strategy since the right hand side of \eqref{eq:mut} depends only on the \CI\ update rule $\psi$. This fact is made precise in the following lemma.
\begin{lemma}[Conditional Independence]
\label{lm:condindep}
Suppose $\lambda$ is a \CI\ strategy profile and $\psi$ is a \CI\ update rule consistent with $\lambda$. 
Let $(g,\mu)=f(\lambda,\psi)$. If every agent $k\neq n$ uses the strategy $g^{k}$ along with the belief system $\mu$, then under any policy $g'^n$ of agent $n$, agent $n$'s belief about the states $X_{1:t}$ for $h^n_t\in\mathcal{H}^n_t$ is given by
\begin{align}
&\mathbb{P}^{g'^n,g^{-n}}(x_{1:t}|h^n_t) = \mu^n_t(h^{n}_t)(x_{1:t})\nonumber\\
=&\mathbf{1}_{\{x^n_{1:t}\}}(h^n_{t}) \prod_{k\neq n}\mu^c_{t}(h^c_{t})(x^k_{1:t}) \text{ for all } x_{1:t} \in \mathcal{X}_{1:t}.
\end{align}
\end{lemma}
\begin{proof}
See Appendix \ref{app:CIB}.
\end{proof}
According to Lemma \ref{lm:condindep}, for any $(g,\mu)=f(\lambda,\psi)$ generated from Lemma \ref{lm:befupdate}, we use 
\begin{align}
\mathbb{P}_{\mu}(x^{-n}_{1:t}|h^c_t):= \mu^n_t(h^{n}_t)(x^{-n}_{1:t})= \mu^c_t(h^{c}_t)(x^{-n}_{1:t})
\end{align}
to indicate that $\mu^n_t(h^{n}_t)(x^{-n}_{1:t})$ depends only on $h^c_t$ and $\mu$.

\subsection{Common Information Based Perfect Bayesian Equilibria}
\label{sub:common}

Based on the concept of \CI\ beliefs and \CI\ strategies, we focus on \CI-PBE defined below.
\begin{definition}[\CI-PBE]
A pair $(\lambda^*,\psi^*)$ of a \CI\ strategy profile $\lambda^*$ and a \CI\ update rule $\psi^*$ is called a Common Information Based Perfect Bayesian Equilibrium (\CI-PBE) if 
$\psi^*$ is consistent with $\lambda^*$ and the pair $(g^*, \mu^*)= f(\lambda^*,\psi^*)$ defined in Lemma \ref{lm:befupdate} forms a PBE.
\end{definition}

The following lemma plays a crucial role in establishing the main results of this paper.

\begin{lemma}[Closeness of \CI\ Strategies]
\label{lm:Commonclose}
Suppose $\lambda$ is a \CI\ strategy profile and $\psi$ is a \CI\ update rule consistent with $\lambda$. If every agent $k\neq n$ uses the \CI\ strategy $\lambda^{k}$ along with the belief generated by $\psi$, then, there exists a \CI\ strategy $\lambda'^{n}$ that is a best response for agent $n$ under the belief generated by $\psi$ at every history $h^n_t\in\mathcal{H}^n_t$ for all $t\in\mathcal{T}$. 
\end{lemma}
\begin{proof}
See Appendix \ref{app:CIB}.
\end{proof}

Lemma \ref{lm:Commonclose} says that the set of \CI\ strategies is closed under the best response mapping. 
Since sequential rationality (Definition \ref{def:seqrational}) requires a strategy profile to be a fixed point under the best response mapping (see \eqref{eq:seqrational}), 
Lemma \ref{lm:Commonclose} allows us to restrict attention to the set of \CI\ strategies to find a fixed point and to search for \CI-PBE. 

%We show that the best response mapping has a fixed point within the set of \CI\ strategies and beliefs; thus, we establish the existence of \CI\ equilibria.
%
%
%\begin{theorem}
%\label{thm:publiccommon}
%The dynamic game defined in Section \ref{sec:model} has at least one \CI\ equilibrium.
%\end{theorem}

Below we provide a sequential decomposition of the dynamic game of Section \ref{sec:model} that enables us to sequentially compute \CI-PBE via dynamic programming.

In order to sequentially compute \CI-PBE we define a stage game for each time $t\in\mathcal{T}$ as follows.
\begin{definition}(Stage Game $G_t$)
Given a set of functions $V_{t+1}=\{V^n_{t+1}:\mathcal{X}^n_t\times\mathcal{B}_t\mapsto \mathbb{R},n\in\mathcal{N}\}$ and a belief update function $\psi_t$,
for any realization 
%$c_t \in\mathcal{C}_t,\pi_{t},\hat\pi_{t}\in \Delta(\mathcal{X}_t)$ 
$b_t=(c_t,\pi_t,\hat\pi_t) \in\mathcal{B}_t$ 
we define the following 
%Bayesian game $G_t(V_{t+1},\psi_t,c_t,\pi_t,\hat\pi_t)$.
Bayesian game $G_t(V_{t+1},\psi_t,b_t)$.
\\
%\textbf{Stage Game $G_t(V_{t+1},\psi_t,c_t,\pi_t,\hat\pi_t)$}
\textbf{Stage Game $G_t(V_{t+1},\psi_t,b_t)$}
\begin{itemize}
\item There are $N$ players indexed by $\mathcal{N}$.
\item Each player $n\in \mathcal{N}$ observes private information $X^n_t\in \mathcal{X}^n_t$; $b_t=(c_t,\pi_t,\hat\pi_t)$ are common information.
\item $X_t = (X^1_t,\dots,X^N_t)$ has a prior distribution $\pi_t$.		

\item Each player $n\in \mathcal{N}$ selects an action $A^n_t \in \mathcal{A}^n_t$. 
\item Each player $n\in \mathcal{N}$ has utility 
				\begin{align}
				&U^n_{G_t(V_{t+1},\psi_t,b_t)} \nonumber\\
				:=&\phi^{n}_t(c_t,X_t,A_t)+V^{n}_{t+1}(X^n_{t+1},B_{t+1}), \text{ where }
				\label{eq:stageutility}
				\\
				&B_{t+1}:=(C_{t+1},\psi_{t}(Y_t,A_t,b_t),\hat\psi_{t}(Y_t,A_t,\hat\pi_{t})).
				\end{align}
\end{itemize}
\end{definition}

If for each $t\in\mathcal{T}$, the functions $V_{t+1}$ are associated with the agents' future utilities, the Bayesian game $G_t(V_{t+1},\psi_t,b_t)$ becomes a stage game at $t$ of the original game defined in Section \ref{sec:model}. 
Therefore, we consider Bayesian Nash equilibria (BNE) of the game $G_t(V_{t+1},\psi_t,b_t)$.
For all $V_{t+1}$ and $\psi_t$, we define, the BNE correspondence as follows
\begin{definition}(BNE Correspondence)
\begin{align}
&BNE_t(V_{t+1},\psi_t) := \nonumber\\
\{ &\lambda_t: \forall 
%\, c_t\in \mathcal{C}_t \text{ and }\pi_{t},\hat\pi_{t} \in \Delta(\mathcal{X}_t) , 
b_t \in \mathcal{B}_t,\lambda_t|_{b_t} \text{ is a BNE of }G_t(V_{t+1},\psi_t,b_t), \nonumber\\
&\text{ where }\lambda^n_t|_{b_t}(x^n_t):=\lambda^n_t(x^n_t,b_t)\, \forall\,n\in\mathcal{N},x^n_t\in\mathcal{X}^n_t
% \tilde\lambda_t=\{\tilde\lambda^n_t: \mathcal{X}^n_t\mapsto \Delta(\mathcal{A}^n_t), n\in\mathcal{N} \} 
%\nonumber\\
%& \tilde\lambda_t \text{ is a BNE of }G_t(V_{t+1},\psi_t,b_t) \text{ where }\forall\,n\in\mathcal{N}
%\nonumber\\
%&\tilde\lambda^n_t(x^n_t)(a^n_t):=\lambda^n_t(x^n_t,b_t)(a^n_t)\quad \forall\,x^n_t,a^n_t
\}.
\label{eq:BNEt}
\end{align}
\end{definition}
If $\lambda_t|_{b_t}$ is a BNE of $G_t(V_{t+1},\psi_t,b_t)$, then for all $n\in\mathcal{N}$ and any realization $x^n_t\in\mathcal{X}^n_t$, 
any $a^n_t\in\mathcal{A}^n_t$ such that $\lambda^n_t|_{b_t}(x^n_t)(a^n_t)>0$ should satisfy
\begin{align}
a^n_t \in \argmax_{a'^n_t \in \mathcal{A}^n_t } \left\{
\mathbb{E}^{\lambda^{-n}_t}_{\pi_{t}} \left[U^n_{G_t(V_{t+1},\psi_t,b_t)}|x^n_t,a'^n_t\right]\right\}.
\end{align}

%In order to sequentially compute \CI\ equilibria, for each $t\in\mathcal{T}$, given a set of functions $V_{t+1}=\{V^n_{t+1}(\cdot),n\in\mathcal{N}\}$ and an update rule $\psi_t$, we define, the Bayesian Nash equilibrium (BNE) correspondence as follows
%\begin{definition}(BNE Correspondence)
%\begin{align}
%BNE_t(V_{t+1},\psi_t) := \{ &\lambda_t: \forall\, n\in \mathcal{N}, x^n_t\in \mathcal{X}^n_t, c_t\in \mathcal{C}_t \text{ and }\pi_{t},\hat\pi_{t} \in \Delta(\mathcal{X}_t) , \nonumber\\
%&\text{ if } \lambda^{n}_t(x^n_t,c_t,\pi_{t},\hat\pi_{t})(a^{n}_t)>0, \text{ then }a^n_t \in \mathcal{A}^n_t
%\text{ should be a solution to \eqref{eq:BRt}} \}
%\label{eq:BNEt}
%\end{align}
%where for any realizations $x^n_t,c_t,\pi_{t},\hat\pi_{t}$ and functions $\lambda_t,\psi_t, V_{t+1}$, \eqref{eq:BRt} is an optimization/best response problem defined by
%\begin{align}
%& a^{n}_t\in
% \argmax_{a'^n_t \in \mathcal{A}^n_t } \left\{
%\mathbb{E}^{\lambda^{-n}_t}_{\pi_{t}} \left[\phi^{n}_t(c_t,X_t,(a'^n_t,A^{-n}_t))\right.\right.\nonumber\\
%&\quad +\left.\left.V^{n}_{t+1}(X^n_{t+1},C_{t+1},\psi_{t}(\pi_{t},\hat\pi_{t},(a'^n_t,A^{-n}_t),Y_t)),\hat\psi_{t}(\hat\pi_{t},(a'^n_t,A^{-n}_t),Y_t))|x^n_t,c_t,\pi_{t},\hat\pi_{t}\right]\right\}.
%\label{eq:BRt}
%\end{align}
%\end{definition}

Similar to the dynamic program in stochastic control, for each time $t\in\mathcal{T}$ we define the value update function $D^n_t(V_{t+1},\lambda_t,\psi_t)$ for each $n\in\mathcal{N}$.
\begin{definition}(Value Update Function)
\begin{align}
&D^n_t(V_{t+1},\lambda_t,\psi_t)(x^n_t,b_t) :=
\mathbb{E}^{\lambda_t}_{\pi_{t}} \left[U^n_{G_t(V_{t+1},\psi_t,b_t)}|x^n_t \right].
\label{eq:valuefun}
\end{align}
%\begin{align}
%D_t(V_{t+1},\lambda_t,\psi_t):= \left\{ V^n_{t}(\cdot),n\in\mathcal{N}:\quad
%V^{n}_{t}(x^n_t,c_t,\pi_{t},\hat\pi_{t}) :=
%\mathbb{E}^{\lambda_t}_{\pi_{t}} \left[U^n_{G_t(V_{t+1},\psi_t,c_t,\pi_t,\hat\pi_t)}|x^n_t \right] \forall\, n\in\mathcal{N}
%\right\}.
%\end{align}
%$V_t :=D(V_{t+1},\lambda_t,\psi_t)$ defines a set of functions $V_t=\{V^n_{t}(\cdot),n\in\mathcal{N}\}$ such that for all $n\in\mathcal{N}$, 
%\begin{align}
%&V^{n}_{t}(x^n_t,c_t,\pi_{t},\hat\pi_{t}) :=
%\mathbb{E}^{\lambda_t}_{\pi_{t}} \left[U^n_{G_t(V_{t+1},\psi_t,c_t,\pi_t,\hat\pi_t)}|x^n_t \right]
%\label{eq:valuefun}
%\end{align}
\end{definition}
If $V^n_t= D^n_t(V_{t+1},\lambda_t,\psi_t)$, for any realization $x^n_t\in\mathcal{X}^n_t$ and $b_t\in\mathcal{B}_t$,
the value $V^{n}_{t}(x^n_t,b_t)$
denotes player $n$'s expected utility under the strategy profile $\lambda_t|_{b_t}$ in game $G_t(V_{t+1},\psi_t,b_t)$.

%In order to compute \CI\ equilibria, for a pair of \CI\ strategy profile $\lambda$ and \CI\ update rule $\psi$ (not necessary consistent with $\lambda$), we define, recursively, a set of functions
%\begin{align}
%&V(\lambda,\psi):=\left\{V_{t}^{n}(\cdot), n\in\mathcal{N},t \in\mathcal{T}\right\}
%\end{align}
%as follows.
%\\
%\begin{align}
%V^{n}_{T+1}(\cdot):=0\quad\forall n\in \mathcal{N}.
%\end{align}
%For each $t \in \mathcal{T}$, for all $n\in \mathcal{N}$,
%\begin{align}
%&V^{n}_{t}(x^n_t,c_t,\pi_{t},\hat\pi_{t}):= \mathbb{E}^{\lambda_t}_{\pi_{t}} \left[\phi^{n}_t(c_t,X_t,A_t)
%+V^{n}_{t+1}(X^n_{t+1},C_{t+1},\psi_{t}(\pi_{t},\hat\pi_{t},A_t,Y_t),\hat\psi_{t}(\hat\pi_{t},A_t,Y_t))|x^n_t,c_t,\pi_{t},\hat\pi_{t} \right]
%\label{eq:valuefun}
%\end{align}
%where $\pi_{t}$ denotes the \CI\ distribution of $X_t$ at $t$ in the above expectation.

Using the concept of stage games and value update functions, we provide a dynamic programming method to sequentially compute \CI-PBE in the following theorem.
\begin{theorem}(Sequential Decomposition)
\label{thm:commonDP}
A pair $(\lambda^*,\psi^*)$ of a \CI\ strategy profile $\lambda^*$ and a \CI\ update rule $\psi^*$ is a \CI-PBE if $(\lambda^*,\psi^*)$ solves the dynamic program for the value functions $V^n_t(\cdot),n\in\mathcal{N},t\in\mathcal{T}\cup\{T+1\}$ defined by \eqref{eq:DP:end}-\eqref{eq:DP:valueupdate} below.
\begin{align}
V^{n}_{T+1}(\cdot):=0\quad\forall n\in \mathcal{N};
\label{eq:DP:end}
\end{align}
for all $t\in\mathcal{T}$
\begin{align}
&\lambda^*_t \in BNE_t(V_{t+1},\psi^*_t), \label{eq:DP:BNE}\\
&\psi^*_t \text{ is consistent with }\lambda^*_t,\label{eq:DP:consistent}\\
&V^n_t =D^n_t(V_{t+1},\lambda^*_t,\psi^*_t)\quad\forall n\in \mathcal{N}. \label{eq:DP:valueupdate}
\end{align}

\end{theorem}
%\begin{theorem}(Sequential Decomposition)
%\label{thm:commonDP}
%A pair $(\lambda^*,\psi^*)$ of a \CI\ strategy profile $\lambda^*$ and a \CI\ update rule $\psi^*$ is a \CI\ equilibrium if and only if $\psi^*$ is consistent with $\lambda^*$ and
%$\lambda^*$ solves the dynamic programs defined by functions $V(\lambda^*,\psi^*)=\left\{V_{t}^{*n}(\cdot), n\in\mathcal{N},t \in\mathcal{T}\right\}$ given by \eqref{eq:optimalS} below.
%That is, for all $t\in \mathcal{T}$, for all $n\in \mathcal{N}$, for any  $x^n_t\in \mathcal{X}^n_t $, $c_t\in \mathcal{C}_t $ and $\pi_{t},\hat\pi_{t} \in \Delta(\mathcal{X}_t) $,
%any $a^{*n}_t \in \mathcal{A}^n_t$, such that $\lambda^{*n}_t(x^n_t,c_t,\pi_{t},\hat\pi_{t})(a^{*n}_t)>0$, should satisfy
%\begin{align}
%& a^{*n}_t\in
% \argmax_{a^n_t \in \mathcal{A}^n_t } \left\{
%\mathbb{E}^{\lambda^{-n}_t}_{\pi_{t}} \left[\phi^{n}_t(c_t,X_t,(a^n_t,A^{-n}_t))\right.\right.\nonumber\\
%&\quad +\left.\left.V^{n}_{t+1}(X^n_{t+1},C_{t+1},\psi_{t}(\pi_{t},\hat\pi_{t},(a^n_t,A^{-n}_t),Y_t)),\hat\psi_{t}(\hat\pi_{t},(a^n_t,A^{-n}_t),Y_t))|x^n_t,c_t,\pi_{t},\hat\pi_{t}\right]\right\}.
%\label{eq:optimalS}
%\end{align}
%
%\end{theorem}

\begin{proof}
See Appendix \ref{app:CIB}.
\end{proof}

Note that, from the dynamic program, the value function $V^n_1(x^n_1,(c_1,\pi_1,\pi_1))$ at time $t=1$ gives agent $n$'s expected utility corresponding to the \CI-PBE $(\lambda^*,\psi^*)$ conditional on his private information $X^n_1=x^n_1$ and public state $C_1 = c_1$ when the prior distribution of $X_1$ is $\pi_1$.

Using Theorem \ref{thm:commonDP}, we can compute \CI-PBE of the dynamic game.
The following algorithm uses backward induction to compute \CI-PBE based on Theorem \ref{thm:commonDP}.
\begin{algorithm}[H]
\caption{Backward Induction for Computing \CI-PBE}
\label{alg:BICIB}
\begin{algorithmic}[1]
%\STATE $\mathcal{V}_{T+1} \leftarrow \{0\}$
\STATE $V_{T+1} \leftarrow 0$, $\lambda^* \leftarrow \emptyset$, $\psi^* \leftarrow \emptyset$
\FOR{$t=T$ to $1$}
%	\STATE $\mathcal{V}_{t} \leftarrow \emptyset$
%	\FOR{every $V_{t+1} \in \mathcal{V}_{t+1}$}
%	\FOR{every $c_t\in \mathcal{C}_t$, $\pi_t,\hat\pi_t\in\Delta(\mathcal{X}_t)$}
	\FOR{every $b_t\in \mathcal{B}_t$}
%		\STATE Construct a Bayesian game $G_t(c_t,\pi_t,\hat\pi_t)$ of $N$ players:
%		\STATE\hspace{\algorithmicindent} The state $X_t = (X^1_t,\dots,X^N_t)$ has a prior distribution $\pi_t$.		
%		\STATE\hspace{\algorithmicindent} Each player $n\in \mathcal{N}$ observes private information $X^n_t\in \mathcal{X}^n_t$ and $c_t,\pi_t,\hat\pi_t$ are common information.
%
%		\STATE\hspace{\algorithmicindent} Each player $n\in \mathcal{N}$ selects an action $A^n_t \in \mathcal{A}^n_t$ 
%		\STATE\hspace{\algorithmicindent} Given a belief update rule $\tilde\psi_{t}:\mathcal{A}_t\times\mathcal{Y}_t \mapsto \Delta(\mathcal{X}_t)$, each player $n\in \mathcal{N}$ has utility 
%				\begin{align*}
%				U^n_{G_t} := \phi^{n}_t(c_t,X_t,A_t)+V^{n}_{t+1}(X^n_{t+1},C_{t+1},\tilde\psi_{t}(A_t,Y_t)),\hat\psi_{t}(\hat\pi_{t},A_t,Y_t))
%				\end{align*}
		\STATE Construct the stage game $G_t(V_{t+1},\psi_t,b_t)$
		\STATE Compute $(\lambda^*_{t}|_{b_t},\psi^*_{t}|_{b_t})$ such that $\psi^*_{t}|_{b_t}$ is consistent with $\lambda^*_{t}|_{b_t}$, and $\lambda^*_{t}|_{b_t}$ is a BNE of $G_t(V_{t+1},\psi^*_t|_{b_t},b_t)$
		\label{eq:findBNE}
%		\STATE $(\tilde\lambda^*_{t},\tilde\psi^*_{t})\leftarrow$ a PBE for game $G_t(c_t,\pi_t,\hat\pi_t)$.
%		\label{eq:findBNE}
%		That is, $\tilde\psi^*_{t}$ is consistent with $\tilde\lambda^*_{t}$, and $\tilde\lambda^{*n}_{t}:\mathcal{X}^n_t \mapsto \Delta(\mathcal{A}^n_t)$ maximizes player $n$'s expected utility given $\tilde\psi^*_{t}$ and $\tilde\lambda^{*}_{t}$. 
		\FOR{every $n\in \mathcal{N}$}
		\STATE $\lambda^{*n}_t(x^n_t,b_t) \leftarrow  \lambda^*_{t}|_{b_t}(x^n_t)$, $x^n_t\in\mathcal{X}^n_t$
		\STATE $\psi^{*n}_t(y_t,a_t,b_t) \leftarrow  \psi^*_{t}|_{b_t}(y_t,a_t)$,  $(y_t,a_t)\in\mathcal{Y}_t\times\mathcal{A}_t$
%		\STATE $V^n_t(x^n_t,c_t,\pi_t,\hat\pi_t) \leftarrow  \mathbb{E}^{\lambda^*_t}_{\pi_{t},\psi^*_{t}}\left[ U^n_{G_t}|x^n_t,c_t\right]$ 
		\STATE $V^n_t(x^n_t,b_t) \!\leftarrow\!  D^n_t(V_{t+1},\lambda^*_t,\psi^*_t)(x^n_t,b_t)$, $x^n_t\in\mathcal{X}^n_t$
		\ENDFOR
	\ENDFOR
%	\STATE $V_{t} \leftarrow \{V^1_t,V^2_t,\dots,V^N_t\}$
	\STATE $\lambda^* \leftarrow (\lambda^*_t,\lambda^*)$,  $\psi^* \leftarrow (\psi^*_t,\psi^*)$
%	\STATE $\leftarrow V^n_1(x^n_1,c_1,\pi_1,\pi_1)$
%		\STATE $\mathcal{V}_{t} \leftarrow \mathcal{V}_{t} \cup\{V_t\}$
%	\ENDFOR
\ENDFOR
\end{algorithmic}
\end{algorithm}
Note that in line \ref{eq:findBNE}, for different $(\lambda^*_{t}|_{b_t},\psi^*_{t}|_{b_t})$ the algorithm will produce different \CI-PBE.
Finding the pair $(\lambda^*_{t}|_{b_t},\psi^*_{t}|_{b_t})$ in line \ref{eq:findBNE} of Algorithm \ref{alg:BICIB} requires solving a fixed point problem to get a BNE along with a consistent belief system. The complexity for this step is the same as the complexity of finding a PBE for a two-stage dynamic game. 

\section{Example: Multiple Access Broadcast Game}
\label{sec:example}
In this section, we illustrate the sequential decomposition developed in Section \ref{sec:CIBequilibria} with an example of a two-agent multiple access broadcast system.

Consider a multiple access broadcast game where two agents, indexed by $\mathcal{N}=\{1,2\}$, share a common collision channel over time horizon $\mathcal{T}$. 
At time $t$, $W^n_t\in\{0,1\}$ packets arrive at each agent $n\in\mathcal{N}$ according to independent Bernoulli processes with 
$\mathbb{P}(W^1_t=1)=\mathbb{P}(W^2_t=1)=p=0.5$.
Each agent can only store one packet in his local buffer/queue.
Let $X^n_t\in \mathcal{X}^n_t=\{0,1\}$ denote the queue length (number of packets) of agent $n$ at the beginning of $t$.
If a packet arrives at agent $n$ when his queue is empty, the packet is stored in agent $n$'s buffer; otherwise, the packet is dropped, and agent $n$ incurs a dropping cost of $c=2$ units.

At each time $t$, agent $n$ can transmit $A^n_t\in\mathcal{A}^n_t=\{0,1\}$ packets through the shared channel. If only one agent transmits, the transmission is successful and the transmitted packet is removed from the queue. If both agents transmit simultaneously, a collision occurs and both collided packets remain in their queues.
We assume that any packet arriving at time $t$, $t\in\mathcal{T}$, can be transmitted after $t$.
%Assuming that at any time $t$, packets arrive at the agents after their transmissions at time $t$.
Then, the queue length processes have the following dynamics. For $n=1,2$
\begin{align}
X^n_{t+1} = \min\left\{X^n_t-A^n_t(1-A^{-n}_t)+W^n_t,1\right\}.
\label{eq:mac_dynamics}
\end{align}

Assume that agents' transmission results at $t$ are broadcast at the end of time $t$. Then agent $n$'s transmission decision $A^n_t$ at time $t$ is made based on his history of observation $H^n_t = (X^n_{1:t},A_{1:t-1})$ that consists of his local queue lengths and all previous transmissions from both agents.

Suppose each agent gets a unit reward at $t$ if there is a successful transmission at $t$. Then, agent $n$'s utility at time $t$ 
is the reward minus the (expected) dropping cost given by
\begin{align}
%U^n_t =&\phi^{n}_t(X_t,A_t) \nonumber\\
%=  &A^n_t\oplus A^{-n}_t- c\,\mathbb{P}(X^n_t-A^n_t(1-A^{-n}_t)+W^n_t > 1| X_t,A_t)
 &U^n_t =\phi^{n}_t(X_t,A_t)= \nonumber\\
  &A^n_t\hspace{-0.3em}\oplus A^{-n}_t \hspace{-0.5em}- c\,\mathbb{P}(X^n_t\hspace{-0.3em}-A^n_t(1\hspace{-0.2em}-A^{-n}_t)+W^n_t > 1| X_t,A_t)
\label{eq:mac_utility}
\end{align}
where $x\oplus y$ denotes the binary XOR operator, and $n\in\mathcal{N}$.

The multiple access broadcast game described above is an instance of the general dynamic model described in Section \ref{sec:model}.
In the following, we use Algorithm \ref{alg:BICIB} developed in Section \ref{sec:CIBequilibria} to compute a \CI-PBE of this multiple access broadcast game for two time periods, i.e. $\mathcal{T}=\{1,2\}$.

Before applying Algorithm \ref{alg:BICIB}, we note some special features of this multiple access broadcast game.
First, there is no $C_t,Y_t$ in this multiple access broadcast game. Second, since the private state $X^n_t$ can take only values in $\mathcal{X}^n_t = \{0,1\}$, any \CI\ belief in $\Delta(\mathcal{X}^n_t)$ can be described by a number $\pi^n_t \in [0,1]$ for all $n=1,2, t=1,2$.
Furthermore, given any realization $b_t=(\pi_t,\hat\pi_t) \in \mathcal{B}_t$, any \CI\ strategy $\lambda^n_t(x^n_t,b_t), x^n_t \in\{0,1\},$ of agent $n$ can be characterized by a number $\beta^n_t\in[0,1]$ where
\begin{align}
%\beta^n_t := \lambda^n_t(1,b_t)(1) = \mathbb{P}^{\lambda^n_t}(A^n_t=1|X^n_t=1,\pi_t,\hat\pi_t).
\beta^n_t := \lambda^n_t(1,b_t)(1).
\end{align}
This is because $A^n_t$ is binary, and 
%$\mathbb{P}^{\lambda^n_t}(A^n_t=0|X^n_t=0,\pi_t,\hat\pi_t)=1$ 
$\lambda^n_t(0,b_t)(1)=0$ 
because no packet can be transmitted from an empty queue.

%$\lambda^n_t(1,\pi_t,\hat\pi_t)(0)=1-\lambda^n_t(1,\pi_t,\hat\pi_t)(1)$ from the fact that there are only two possible actions, and $\lambda^n_t(0,\pi_t,\hat\pi_t)(1)=1-\lambda^n_t(0,\pi_t,\hat\pi_t)(1)=0$ since no packet can be transmitted from an empty queue.

We now use Algorithm \ref{alg:BICIB} to sequentially compute a \CI-PBE of the multiple access broadcast game.
\\\underline{Construction of the stage game at $t=2$}

At $t=2$, for any $b_2=(\pi_2,\hat\pi_2) \in \mathcal{B}_2$, we construct the stage game $G_2(b_2)$ which is a Bayesian finite game (no need to consider a \CI\ update function because this is the last stage).
\\\underline{Computation of BNE at $t=2$}

 Using standard techniques for static games, we obtain a BNE of $G_2(b_2)$ that is characterized by $\beta^*_2(b_2)= (\beta^{*1}_2(b_2),\beta^{*2}_2(b_2))$, and $\beta^*_2(b_2)$ is given by
\begin{align}
\beta^*_2(b_2) = \left\{
\begin{array}{ll}
(1,1) & \text{ if } \pi^1_2,\pi^2_2< c^*,\\
(0,1) & \text{ if } \pi^1_2< c^*,\pi^2_2 \geq c^*,\\
(1,0) & \text{ if } \pi^1_2\geq c^*,\pi^2_2< c^*,\\
(\frac{c^*}{\pi^1_2},\frac{c^*}{\pi^2_2}) & \text{ if } \pi^1_2,\pi^2_2\geq c^*,
\end{array}
\right.
\end{align}
where $c^*:= \frac{1+cp}{2+cp}$. Then we obtain a \CI\ strategy $\lambda^{*n}_2(1,b_2)(1)=\beta^{*n}_2(b_2)$ for $n=1,2$ at time $t=2$.
\\\underline{Value functions' update at $t=2$}

%The corresponding value functions 
$V^n_2(x^n_2,b_2)=D^n_2(\lambda^*_2)(x^n_2,b_2), n=1,2,$ 
%$V^n_2=D^n_2(\lambda^*_2), n=1,2,$ 
are given by
\begin{align}
&\hspace{-0.5em} V^n_2(1,b_2) \!=\! \left\{\hspace{-0.5em}
\begin{array}{ll}
1 \!-\! \pi^{-n}_2(1+cp) \hspace{-0.5em}& \text{if } \pi^1_2,\pi^2_2< c^*,\\
\pi^{-n}_2 -cp & \text{if } \pi^n_2< c^*,\pi^{-n}_2 \geq c^*,\\
1 & \text{if } \pi^{n}_2\geq c^*,\pi^{-n}_2< c^*,\\
c^* - cp & \text{if } \pi^1_2,\pi^2_2\geq c^*.
\end{array}
\right.
\label{eq:ValueN1}
\\
&\hspace{-0.5em} V^n_2(0,b_2) = \left\{\hspace{-0.5em}
\begin{array}{ll}
\pi^{-n}_2 & \text{ if } \pi^1_2,\pi^2_2< c^*,\\
\pi^{-n}_2& \text{ if } \pi^n_2< c^*,\pi^{-n}_2 \geq c^*,\\
0 & \text{ if } \pi^{n}_2\geq c^*,\pi^{-n}_2< c^*,\\
c^*  & \text{ if } \pi^1_2,\pi^2_2\geq c^*.
\end{array}
\right.
\label{eq:ValueN0}
\end{align}
\\\underline{Construction of the stage game at $t=1$}

At $t=1$, for any $b_1=(\pi_1,\hat\pi_1)\in\mathcal{B}_1$ and a \CI\ update function $\psi_1$, we construct the stage game $G_1(V_2,\psi_1,b_1)$ such that each player $n, n=1,2,$ has utility 
\begin{align}
& U^n_{G_1(V_2,\psi_1,b_1)} \nonumber\\
=& \phi^{n}_1(X_1,A_1) + V^{n}_{2}(X^n_{2},(\psi_{1}(A_1,b_1)),\hat\psi_{1}(A_1,\hat\pi_{1}))).
\label{eq:macutility1}
\end{align}
\iflongversion
\else	%short version
\fi
\\\underline{Computation of BNE and belief update function at $t=1$}

When the players use \CI\ strategies $\lambda_1$ characterized by $\beta_1=(\beta^1_t,\beta^2_1)$, from \eqref{eq:tilpsiS} and \eqref{eq:tilpsiSD0}, we obtain a \CI\ update function $\psi_1$, given below, that is consistent with $\lambda_1$ (we select $\psi^{n}_1(a_1,b_1)\! =\!\hat\psi^{n}_1(a_1,\hat\pi_{1})$ when the denominator of \eqref{eq:tilpsiS} is zero).
\begin{align}
\hspace{-0.5em}\psi^n_{1}(a_t,b_1) = \left\{
\begin{array}{ll}
1 & \text{ if } a^n_t=1, a^{-n}_t=1,\\
p & \text{ if } a^n_t=1, a^{-n}_t=0,\\
\frac{p+\pi^n_1(1-p-\beta^n_1)}{1-\pi^n_1\beta^n_1} & \text{ if } a^n_t=0.
\end{array}
\right.
\label{eq:mac_beliefupdate}
\end{align}
\iflongversion
%Substituting \eqref{eq:mac_beliefupdate} into \eqref{eq:mac_condition1} and \eqref{eq:mac_condition0}, 
Substituting \eqref{eq:mac_beliefupdate} into \eqref{eq:macutility1}, 
\else	%short version
Substituting \eqref{eq:mac_beliefupdate} into \eqref{eq:macutility1}, 
\fi
we have the utilities of the two players in game $G_1(V_2,\psi_1,b_1)$. We numerically compute a BNE of $G_1(V_2,\psi^*_1,b_1)$, characterized by $\beta^*_1(b_1)=(\beta^{*1}_1(\pi_1),\beta^{*2}_1(\pi_1))$ such that $\psi^*_1$ satisfies \eqref{eq:mac_beliefupdate} when $\beta_1 = \beta^*_1$.
The values of $\beta^{*1}_1(\pi_1)$ and $\beta^{*2}_1(\pi_1)$ are shown in Fig. \ref{fig:beta1} for different $\pi_1 \in \Delta(\mathcal{X}_t)=[0,1]\times[0,1]$.

\begin{figure}[h]
\centering
\includegraphics[width=1\textwidth]{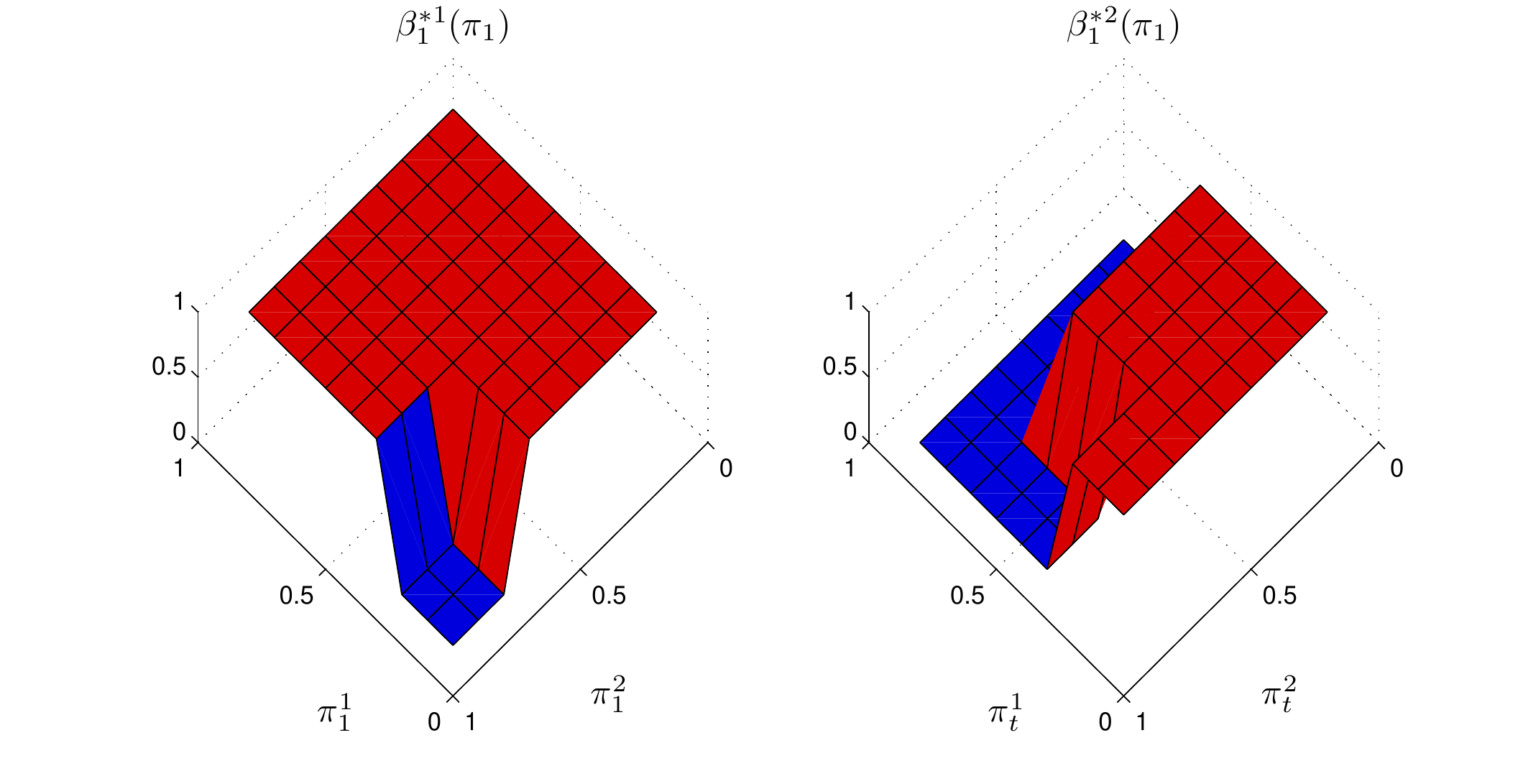}
\caption{Strategies $\beta^{*1}_1(\pi_1)$ and $\beta^{*2}_1(\pi_1)$ in the stage game at time $t=1$.}
\label{fig:beta1}
\end{figure}
Then, we obtain a \CI\ strategy $\lambda^{*n}_1(1,b_1)(1)=\beta^{*n}_1(\pi_1)$ for $n=1,2$ at time $t=1$.
\\
\underline{\CI-PBE and agents' expected utilities}

From the above computation at $t=1,2$, we obtain a \CI-PBE $(\lambda^*,\psi^*)$, where $\lambda^* = (\lambda^*_1,\lambda^*_2)$ and $\psi^* = (\psi^*_1,-)$.

\iflongversion
Using \eqref{eq:macutility1}, 
%Using \eqref{eq:mac_condition1} and \eqref{eq:mac_condition0}, 
\else % short version
Using \eqref{eq:macutility1}, 
\fi
we numerically compute the value functions $V^1_1(x^1_1,b_1)= V^1_1(x^1_1,\pi_1)=D^1_1(V_2,\lambda^*_1,\psi^*_1)(x^1_1,b_1)$ for $x^1_t=0,1,$ for agent $1$. The results are shown in Fig. \ref{fig:values}. These value functions give agent $1$'s (conditional) expected utilities in the \CI-PBE $(\lambda^*,\psi^*)$. Agent $2$'s expected utilities can be computed in a similar way.
\begin{figure}[h]
\centering
\includegraphics[width=1\textwidth]{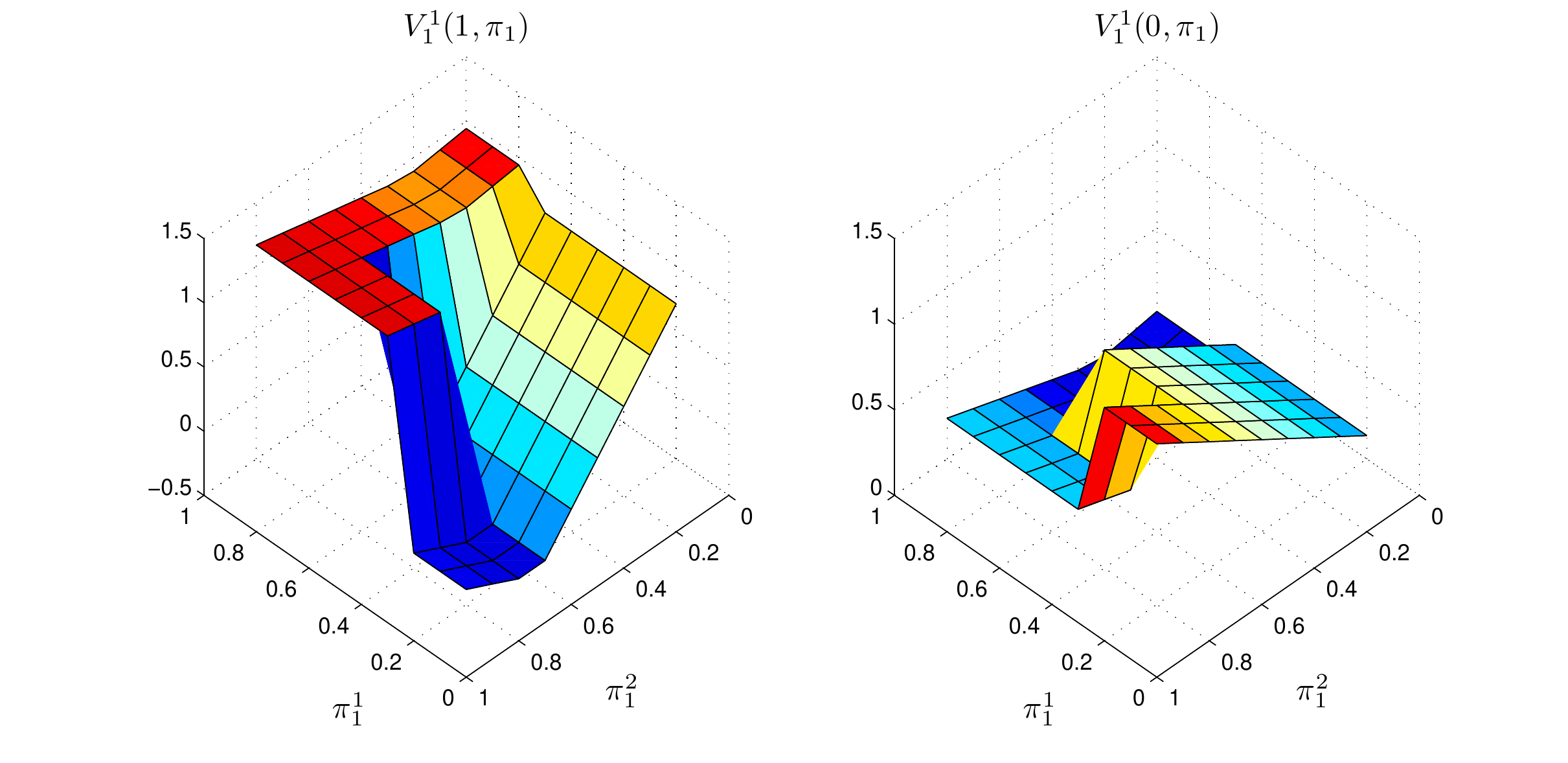}
\caption{Agent $1$'s expected utility $V^n_1(x^n_1,\pi_1)$ in the \CI-PBE $(\lambda^*,\psi^*)$.}
\label{fig:values}
\end{figure}

\begin{remark}
\label{rm:value_discontinuity}
The results show that agents' beliefs depend on their strategies (see \eqref{eq:mac_beliefupdate}). Therefore, there is signaling in this multiple access broadcast game.
Moreover, the value functions are discontinuous in the agents' beliefs
(see \eqref{eq:ValueN1} and \eqref{eq:ValueN0} for time $t=2$, and  Fig. \ref{fig:values} for time $t=1$). The presence of signaling together with the discontinuity of value functions make the agents' utilities discontinuous in their (behavioral) strategies.

\end{remark}

\section{Existence of Common Information Based Perfect Bayesian Equilibria}
\label{sec:existence}
We prove the existence of a \CI-PBE for a subclass of the dynamic games described in section \ref{sec:model}. This subclass includes dynamic games with uncontrolled dynamics and no private values. No private values simply means that each agent's private information $X_t^n$ is payoff irrelevant to himself, but possibly payoff relevant to the other agents. The classic cheap-talk game \cite{fudenberg1991game} and the problem considered in \cite{ouyang2015dynamic} are examples of this subclass. We conjecture that there always exists a \CI-PBE for the general model described in Section \ref{sec:model}. We discuss this conjecture and elaborate more on the difficulty of establishing the existence proof for the general model of Section \ref{sec:model} at the end of this section.

To proceed formally, let \textbf{Game M} denote a dynamic game with uncontrolled dynamics, no private values, finite action spaces $\mathcal{A}_t^n$, $n\in\mathcal{N}$, $t\in\mathcal{T}$, and (possibly) sequential moves. Let $\overline{\mathcal{T}}:=\{t_1,t_2,\cdots,t_K\}\subset \mathcal{T}$ denote the set of time instants in which the system evolves according to the following uncontrolled dynamics

\begin{align}
\hspace{-0.5em}X_{t+1}^n\!=\! \left\{\hspace{-0.5em} \begin{array}{ll}
X^n_t & \text{if } t \neq t_k \text{ for all } t_k\in\overline{\mathcal{T}},\\
f_{t_k}^n(X_{t_k}^n,W_{t_k}^{n,X}) & \text{if } t = t_k \text{ for } t_k\in\overline{\mathcal{T}}.
\end{array} \right.
%X_{t+1}^n=\left\{ \begin{array}{ll}
%X^n_t & \text{if } t \neq t_k \text{ for all } t_k\in\overline{\mathcal{T}}\\
%f_{t_k}^n(X_{t_k}^n,W_{t_k}^{n,X}) & \text{if } t = t_k \text{ for } t_k\in\overline{\mathcal{T}}
%\end{array} \right.
\end{align}
At $t_k<t\leq t_{k+1}$ agents make decisions sequentially in $t_{k+1}-t_k$ epochs. We assume that the order according to which the agents take decisions is known a priori. Furthermore, agents observe the other agents' decisions in previous epochs; this fact is captured by including/appending previous actions in the common state $C_t$ as follows
\begin{align}
%&C_{t+1} =
%\left\{ \begin{array}{ll}
%(C_{t},A_{t}) & \text{if } t \neq t_{k}\text{ for all } t_k\in\overline{\mathcal{T}},\\
%f_{t_k}^c(C_{t_{k-1}+1},W_{t_k}^{C}) & \text{if } t = t_k \text{ for } t_k\in\overline{\mathcal{T}}.
%\end{array} \right.
\hspace{-0.5em}C_{t+1} \!=\!
\left\{\hspace{-0.5em}\begin{array}{ll}
(C_{t},A_{t}) & \text{\!\!if } t \neq t_{k}\text{ for all } t_k\in\overline{\mathcal{T}},\\
f_{t_k}^c(C_{t_{k-1}+1},W_{t_k}^{C}) & \text{\!\!if } t = t_k \text{ for } t_k\in\overline{\mathcal{T}}.
\end{array} \right. 
\end{align}
The agents have a common observation $Y_{t_k}$ at each time $t_k\in\overline{\mathcal{T}}$ when the system evolves.
The observations $Y^n_t, n\in\mathcal{N}, t\in \mathcal{T}$ are described by
\begin{align}
\hspace{-0.5em}&Y_{t}^n\!=\!\left\{ \begin{array}{ll}
\text{empty} & \text{\!if } t \neq t_k \text{ for all } t_k\in\overline{\mathcal{T}},\\
h_{t_k}^n(X_{t_k}^n,W_{t_k}^{n,Y}) & \text{\!if } t = t_k \text{ for } t_k\in\overline{\mathcal{T}}.
\end{array} \right.
%&Y_{t}^n\!=\!\left\{ \hspace{-0.5em}\begin{array}{ll}
%\text{empty} & \text{ if } t \neq t_k \text{ for all } t_k\in\overline{\mathcal{T}}\\
%h_{t_k}^n(X_{t_k}^n,W_{t_k}^{n,Y}) & \text{ if } t = t_k \text{ for } t_k\in\overline{\mathcal{T}}
%\end{array} \right.
\end{align}

Agent $n$, $n\in\mathcal{N}$ has instantaneous utility 
\begin{align}
U^n_t = \phi^n_t(C_t,X^{-n}_t,A_t).
\label{eq:G2U}
\end{align}
for time $t$, $t\in\mathcal{T}$. Thus, each agent $n\in\mathcal{N}$ has no private values, hence his private information $X_t^n$ is payoff irrelevant.

From the above description, it is evident that \textbf{Game M} is indeed a subclass of the class of dynamic games described by the model of Section \ref{sec:model}.  The dynamic oligopoly game presented in \cite{ouyang2015dynamic} is an instance of \textbf{Game M}.

The main result of this section is stated in the theorem below and asserts the existence of a \CI-PBE in \textbf{Game M}.
\begin{theorem}
\label{thm:publiccommon}
\textbf{Game M} described in this section has a \CI-PBE which is a solution to the dynamic program defined by \eqref{eq:DP:end}-\eqref{eq:DP:valueupdate} in Theorem \ref{thm:commonDP}.
% $(\lambda^*,\psi^*)$. 
%\red{In particular, for any $t\in\mathcal{T}$, $c_t\in\mathcal{C}_t$, $\pi_t,\hat\pi_t \in \Delta(\mathcal{X}_t)$,
%\begin{align}
%&\psi^{n*}_t(\pi_{t},\hat\pi_t,a_t,y^n_t) = \hat\psi^n_t(\hat\pi_{t},y^n_t)\quad \forall y^n_t \in \mathcal{Y}^n_t,a_t \in \mathcal{A}_t
%\label{eq:publicpsi}
%\\
%&\lambda_t^{*n}(x^n_t,c_t,\pi_t,\hat\pi_t)(a^n_t) = \lambda_t^{*n}(c_t,\pi_t,\hat\pi_t)(a^n_t) \quad \forall x^n_t \in \mathcal{X}^n_t,a^n_t \in \mathcal{A}^n_t.
%\label{eq:publiclambda}
%\end{align}
%}
\end{theorem}

\begin{proof}
See Appendix \ref{app:existence}.
\end{proof}

The proof of Theorem \ref{thm:publiccommon} is constructive. We construct an equilibrium for \textbf{Game M} in which agents use non-private strategies and have signaling-free beliefs which are consistent with the non-private strategy profile.

%Note that the \CI\ equilibrium satisfying \eqref{eq:publicpsi} and \eqref{eq:publiclambda} in Theorem \ref{thm:publiccommon} has a special structure. A strategy satisfying \eqref{eq:publiclambda} is not a signaling strategy because \eqref{eq:publiclambda} assigns the same probability on the actions given different (payoff relevant) private information.
%In such equilibrium, the agents can use the signaling-free beliefs for performance evaluation since the signaling-free beliefs are consistent with the strategy profile without signaling.

There are three reasons why \textbf{Game M} has a \CI-PBE with non-private strategies. 
First, the instantaneous utility $U^n_t = \phi^n_t(C_t,X^{-n}_t,A_t)$ of agent $n,n\in\mathcal{N}$ does not depend on his private information. Therefore, the agent's best response is the same for all realizations of his private information, and a private strategy doest not provide any advantage in terms of higher instantaneous utility. Second, the system dynamics are strategy independent. Therefore, an agent cannot affect the evolution of the system by using a strategy that depends on his private information about the state of the system.
Third, any private strategy does not provide any advantage to an agent in terms of his utility if it can not affect other agents' decisions, and this is the case when all agents use the signaling-free beliefs.

%As we showed before, the \CI-PBE equilibrium introduced in this paper is a PBE, and also a sequential equilibrium. 
As we showed before, the \CI-PBE introduced in this paper are PBE. 
It is known that for finite dynamic games there always exists one sequential equilibrium, and therefore one PBE \cite{fudenberg1991game,osborne1994course}. The proof of existence of sequential equilibria is indirect; it is done by showing the existence of a trembling hand equilibrium \cite{fudenberg1991game,osborne1994course} which is also a sequential equilibrium. The proof of existence of trembling hand equilibrium follows the standard argument in game theory. It uses a suitable fixed point theorem to show the existence of a trembling equilibrium in an equivalent agent-based model representation \cite{fudenberg1991game,osborne1994course}.

There are some technical difficulties in establishing the existence of a \CI-PBE for the general game model considered in this paper. The standard argument in using fixed point theorems is applicable to finite games where the expected utilities are continuous in the agent's mixed strategies. In each stage game arising in the sequential decomposition, say the game at stage $t$, agent $n$'s expected utility (see \eqref{eq:stageutility}) depends on the functions $\{V^n_{t+1},n\in\mathcal{N}\}$. However, the function $V^n_{t+1}$ is not always continuous in the strategies of agent $n$. 
%The example in \cite{ouyang2015dynamic} shows an instant where the value function $V_t^n(\cdot)$ is discontinuous in agents' strategies. 
(see Remark \ref{rm:value_discontinuity} for the multiple access broadcast game in Section \ref{sec:example} and the example in \cite{ouyang2015dynamic}). 
Therefore, the standard argument for establishing the existence of an equilibrium fails for our general model. Even though we can not prove the existence of a \CI-PBE equilibrium, we conjecture that there always exists a \CI-PBE for the general dynamic game described in this paper. 

We note that for the problem formulated in Section \ref{sec:model}, $\{X_t^t,C_t,\Pi_t,\hat{\Pi}_t\}$, $t\in\mathcal{T}$, are sufficient statistics from the decision making point of view (i.e. control theory). This makes a \CI\ strategy a more natural strategy choice for an agent, and consequently, a \CI-PBE is a more plausible equilibrium to arise in practice. However, this does not imply that from the game theory point of view, at all equilibria agents' best responses can be generated using only  $\{X_t^t,C_t,\Pi_t,\hat{\Pi}_t\}$. In a game problem, agents can incorporate a payoff irrelevant information in their strategy choice as a coordination instrument. %\footnote{The concept of correlated equilibrium can be seen as using an information as a coordination instrument that is not part of the game.} 
For example, consider the classic repeated prisoner's dilemma game. In this game, agents can use previous outcomes of the game, that are payoff irrelevant, to sustain a punishment mechanism that results in additional equilibria beyond the repetition of the stage-game equilibrium \cite{mailath2006repeated}. The indirect proof for existence of sequential equilibria and PBE (described above) allows for this type of equilibria that depend on payoff irrelevant information for coordination. Nevertheless, we conjecture that there always exists an equilibrium for the game described in Section \ref{sec:model} that depends only on $\{X_t^t,C_t,\Pi_t,\hat{\Pi}_t\}$. The example of the dynamic multiple access broadcast game in Section \ref{sec:example} is an instance of a dynamic game that does not belong to the subclass of \textbf{Game M}, but has a \CI-PBE. To make our conjecture more precise, we provide below a sufficient condition for the existence of a \CI-PBE.

Let $(g^*,\mu^*)$ be a strategy profile that is a PBE. Consider the following condition. 

\textbf{Condition C:} For all $h_t^c,h_t^{'c}\in\mathcal{H}^c_t$ such that
\begin{align}
 \mathbb{P}_{\mu^*}(x_t|h_t^c)= \mathbb{P}_{\mu^*}(x_t|h_t^{'c}), \forall x_t\in \mathcal{X}_t,\forall n\in\mathcal{N},\forall t\in\mathcal{T} ,
\end{align}
we have
 \begin{align}
 g^{n*}_t(x_{1:t}^n,h_t^c)=g^{n*}_t(x_{1:t}^n,h_t^{c'}),
 \end{align}
 for all $x_{1:t}^n\in\mathcal{X}^n_{1:t}$.

If \textbf{Condition C} is satisfied for a PBE $(g^*,\mu^*)$, then a \CI-PBE exists.
We conjecture that \textbf{Condition C} is satisfied for at least a PBE in the dynamic games considered in this paper.

%\begin{itemize}
%
%\item The dynamic game \textbf{G2} can be viewed as a dynamic game played by multiple agents who can privately observe the state of an uncontrolled subsystem. The system states evolve at each time instant $t_k,k=1,2,\dots$ and the agents can taken their actions either simultaneously or sequentially before the system states evolve at each $t_k,k=1,2,\dots$.
%
%\item Since the system states do not depend on users' actions, 
%
%
%\item There are some technical difficulties to establish the existence of a \CI\ equilibrium in the general game model described in \ref{sec:model}. The general approach to prove the existence of an equilibrium is to apply a suitable fix point theorem. Many fix point theorems are applicable to games where utility functions are continuous in the mixed strategies of the users. However, in each step of our decomposition \eqref{eq:optimalS}, the user $n$'s utilities depends on the functions $V^n_{t+1},n\in\mathcal{N}$. This function $V^n_{t+1}$ is not continuous in the mixed strategies of user $n$ in general. The example in \cite{ouyang2015dynamic} shows an instant where the utility function is not continuous in the users' strategies. Even though we can not prove the existence of a \CI\ equilibrium for the dynamic game described in \ref{sec:model}, we conjecture that there always exists a \CI\ equilibrium.
%
%
%\end{itemize}

\section{Conclusion}
\label{sec:conclusion}
We studied a general class of stochastic dynamic games with asymmetric information. %In the presence of information asymmetry agents form belies about other agents' private information. Signaling occurs when agents reveal information about their private information through their action and affect the others' beliefs. 
We identified game environments that can lead to signaling in dynamic games. We considered a state space model to analyze dynamic games with private states and controlled Markovian dynamics. We provided a comparison between our state space model and classic extensive game form model. We showed that the two models are equivalent as long as one ensures that the belief system associated with the state space model is compatible with the system dynamics. To ensure the compatibility, we introduced the signaling-free belief system. Using the signaling-free belief system, we provided a formal definition of PBE in our state space model. We used the common information among agents to define a subset of PBE called \CI-PBE that consist of a pair of \CI\ strategy profile and \CI\ update rule. We obtained a sequential decomposition for dynamic games that leads to a backward induction algorithm for the computation of \CI-PBE even when signaling occurs. We illustrated our result with an example of multiple access broadcast game. We proved the existence of \CI-PBE for a subclass of dynamic games and provided a sufficient condition for the existence of \CI-PBE for the general class of dynamic games considered in this paper.

%When signaling occurs in dynamic games, agents' beliefs depend on strategy profile. Thus, at a PBE there is a circular connection between equilibrium strategy profile and equilibrium belief system through consistency and sequential rationality conditions. As a result, computation of PBE in such games is a challenging problem. 

\appendices
\section{}
\label{app:CIB}
\begin{proof}[Proof of Lemma \ref{lm:SFupdate}]
The lemma is proved by induction. Since the initial states are independent, \eqref{eq:SFindep} holds at $t=1$.
Suppose the lemma is true at time $t$.
Given any $h^c_{t+1}=\{c_{1:t+1},y_{1:t},a_{1:t}\}$ at $t+1$, we have from Bayes' rule
\begin{align}
&\hat \pi_{t+1}(x_{t+1})  \nonumber\\
= & \mathbb{P}_{(A_{1:t}=a_{1:t})}(x_{t+1}|y_{1:t}) \nonumber\\
= & \frac{\mathbb{P}_{(A_{1:t}=a_{1:t})}(x_{t+1},y_t|y_{1:t-1})}{\sum_{x'_{t+1}\in \mathcal{X}_{t+1}}\mathbb{P}_{(A_{1:t}=a_{1:t})}(x'_{t+1},y_t|y_{1:t-1})}.
\label{eq:SFbef1}
\end{align}
%The last equality in \eqref{eq:SFbef1} follows by the Bayes' rule and the fact that $A_t=a_t$ in the computation of the signaling free belief.
The numerator in \eqref{eq:SFbef1} can be further expressed by
\begin{align}
& \mathbb{P}_{(A_{1:t}=a_{1:t})}(x_{t+1},y_t|y_{1:t-1})
\nonumber\\
=& \sum_{x_t \in \mathcal{X}_t}\mathbb{P}_{(A_{1:t}=a_{1:t})}(x_{t+1},y_t,x_t|y_{1:t-1})
\nonumber\\
\stackrel{(a)}{=}& \sum_{x_t\in \mathcal{X}_t} \prod_{n=1}^N p^n_t(x^n_{t+1};x^n_t,a_t) q^n_t(y^n_t;x^n_t,a_t) \hat\pi^n_t(x^n_t)
\nonumber\\
=& \prod_{n=1}^N\sum_{x^n_t\in \mathcal{X}^n_t} p^n_t(x^n_{t+1};x^n_t,a_t) q^n_t(y^n_t;x^n_t,a_t) \hat\pi^n_t(x^n_t).
\label{eq:SFbef2}
\end{align}
Equation (a) in \eqref{eq:SFbef2} follows from the system dynamics, the fact $A_t=a_t$ and the induction hypothesis for the lemma.
Substituting \eqref{eq:SFbef2} into both the numerator and denominator of \eqref{eq:SFbef1} we get
\begin{align}
&\hat \pi_{t+1}(x_{t+1})  \nonumber\\
%= & \frac{\prod_{n=1}^N\sum_{x^n_t\in \mathcal{X}^n_t} p^n_t(x^n_{t+1};x^n_t,a_t) q^n_t(y^n_t;x^n_t,a_t) \hat\pi^n_t(x^n_t) }
%{\sum_{x'_{t+1}\in \mathcal{X}_{t+1}}\prod_{n=1}^N\sum_{x'^n_t\in \mathcal{X}^n_t} p^n_t(x'^n_{t+1};x'^n_t,a_t) q^n_t(y^n_t;x'^n_t,a_t) \hat\pi^n_t(x'^n_t)}
%\nonumber\\
= & \prod_{n=1}^N\frac{\sum_{x^n_t\in \mathcal{X}^n_t} p^n_t(x^n_{t+1};x^n_t,a_t) q^n_t(y^n_t;x^n_t,a_t) \hat\pi^n_t(x^n_t) }
{\sum_{x'^n_t\in \mathcal{X}^n_t} q^n_t(y^n_t;x'^n_t,a_t) \hat\pi^n_t(x'^n_t)}
\nonumber\\
= & \prod_{n=1}^N\hat\psi^n_t(y^n_t,a_t,\hat \pi_{t})(x^n_{t+1}) .
\end{align}

\end{proof}

\begin{proof}[Proof of Lemma \ref{lm:befupdate}]
If  $\lambda$ is a \CI\ strategy profile and $\psi$ is a \CI\ update rule consistent with $\lambda$, 
we define $g\in \mathcal{G}$ to be the strategy profile constructed by \eqref{eq:getaS} from $(\lambda,\psi)$.
%a strategy profile where for all $n\in\mathcal{N}$
%\begin{align}
%&g_t^{n}(h^{n}_t) := \lambda_t^{n}(x^n_t,c_t,\gamma_{\psi,t}(h^c_t),\hat \gamma_t(h^c_t)).
%\label{eq:CIBtostrategy}
%%= \lambda_t^{n}(x^n_t,c_t,\pi^{\gamma_\psi}_t,\hat \pi_t).
%\end{align}

We proceed to recursively define a belief system $\mu$ and maps $\{\mu^c_t,t\in\mathcal{T}\}$ that satisfy \eqref{eq:mut}-\eqref{eq:muct}, and are such that $\mu$ is consistent with $g$.
%For that matter, we include the belief system $\mu^c_t : \mathcal{H}^c_t \mapsto \Delta(X)_{1:t}$ on the common histories $\mathcal{H}^c$. 
For that matter, we first define the signaling-free belief $\hat \mu^c_t: \mathcal{H}^c_t \mapsto \Delta(\mathcal{X}_{1:t})$ given $h^c_t = \{a_{1:t-1},y_{1:t-1}\} \in \mathcal{H}^c_t$ such that for any $x_{1:t} \in \mathcal{X}_{1:t}$
\begin{align}
\hat \mu^c_t(h^c_t)(x_{1:t}) := \mathbb{P}_{(A_{1:t-1}=a_{1:t-1})}(x_{1:t}|y_{1:t-1}).
%\label{eq:sigfreebef}
\end{align}
%To construct the belief system, we first define the natural belief system $\hat \mu$ on $\mathcal{H}^c$ as
%\begin{align}
%\hat \mu(h^c_t)(x_{1:t}) := \mathbb{P}^{g^{a_{1:t-1}}}(x_{1:t}|y_{1:t-1},a_{1:t-1}),
%\end{align}
%where $g^{a_{1:t-1}}$ is defined in \eqref{eq:gfixeda}. 

At time $t=1$ we define, for all $h^n_1=(x^n_1,h^c_1)\in\mathcal{H}^n_1, n\in\mathcal{N}$ and for all $x_1 \in \mathcal{X}_1$
\begin{align}
 \mu^c_1(h^c_1)(x_{1}) :=  &\mathbb{P}(x_1)
\label{eq:muatone2}
\\
 \mu^n_1(h^n_1)(x_{1}) :=  &\mathbf{1}_{\{x^n_1|h^n_1\}} \mathbb{P}(x^{-n}_1)
\label{eq:muatone}
\end{align}
%for any histories $h^n_1, n\in \mathcal{N}$, and $h^c_1$.
Then, \eqref{eq:mut} and \eqref{eq:muct} are satisfied at time $1$, and $g$ is consistent with $\mu$ before time $1$. (basis of induction)

Suppose  $\mu^c_t(h^c_t)(\cdot)$ and $\mu^n_t(h^n_t)(\cdot)$ are defined, \eqref{eq:mut} and \eqref{eq:muct} are satisfied at time $t$, and $g$ is consistent with $\mu$ before time $t$ (induction hypothesis).

We proceed to define $\mu^c_{t+1}(h^c_{t+1})(\cdot)$, and $\mu^n_{t+1}(h^n_{t+1})(\cdot)$,
and prove that \eqref{eq:mut} and \eqref{eq:muct} are satisfied at time $t+1$, and $g$ is consistent with $\mu$ before time $t+1$.
We first define 
\begin{align}
&\eta^k_t(x^k_{t},y^k_t,a_t,h^c_t)
:=\eta^k_t(x^k_{t},y^k_t,a_t,(c_t,\gamma_{\psi,t}(h^c_t),\hat \gamma_t(h^c_t)))\nonumber\\
= &q^k_t(y^k_{t};x^k_t,a_t)\lambda^{k}_t(x^n_{t},c_t,\gamma_{\psi,t}(h^c_t),\hat \gamma_t(h^c_t))(a^n_t).
\label{eq:etainproof}
\end{align}

At time $t+1$, for any histories $h^c_{t+1}$ and $h^n_{t+1}, n\in\mathcal{N}$, we define the beliefs
\begin{align}
&\mu^c_{t+1}(h^c_{t+1})(x_{1:t+1}) :=  \prod_{k\in \mathcal{N}}\mu^c_{t+1}(h^c_{t+1})(x^k_{1:t+1}),
\label{eq:mutupdatec1}
\\
&\mu^n_{t+1}(h^n_{t+1})(x_{1:t+1}) \nonumber\\
%:=& \mathbf{1}_{\{x^n_{1:t+1}\}}(h^n_{t+1}) \prod_{k\neq n}\mu^n_{t+1}(h^n_{t+1})(x^k_{1:t+1}),
:=& \mathbf{1}_{\{x^n_{1:t+1}\}}(h^n_{t+1}) \prod_{k\neq n}\mu^c_{t+1}(h^c_{t+1})(x^k_{1:t+1}),
\label{eq:mutupdate1}
\end{align}
%where for any $k \neq n $
%\begin{align}
%&\mu^n_{t+1}(h^n_{t+1})(x^k_{1:t+1}) 
%:=  \mu^c_{t+1}(h^c_{t+1})(x^k_{1:t+1}); \label{eq:mutupdate}
%\end{align}
%and for for any $k\in\mathcal{N}$
where for for any $k\in\mathcal{N}$
\begin{align}
&\mu^c_{t+1}(h^c_{t+1})(x^k_{1:t+1}) 
\nonumber\\
:= &
\frac{p^k_t(x^k_{t+1};x^k_t,a_t) \eta^k_t(x^k_{t},y^k_t,a_t,h^c_t)\mu^c_{t}(h^c_{t})(x^{k}_{1:t})}
	{\sum_{x'^k_{t} \in \mathcal{X}^k_t}\eta^k_t(x'^k_{t},y^k_t,a_t,h^c_t)  \mu^c_t(h^c_{t})(x'^{k}_t)} 
\label{eq:mutupdatec_nonzero}
\end{align}
when the denominator of \eqref{eq:mutupdatec_nonzero} is non-zero;
when the denominator of \eqref{eq:mutupdatec_nonzero} is zero, $\mu^c_{t+1}(h^c_{t+1})(x^k_{1:t+1}) $ is defined by
\begin{align}
&\mu^c_{t+1}(h^c_{t+1})(x^k_{1:t+1}) 
\nonumber\\
:= &
\left\{\begin{array}{ll}
0 \quad\quad\text{ when }\hat \mu^c_{t+1}(h^c_{t+1})(x^k_{1:t+1}) =0,
\\
\frac{\gamma_{\psi,t+1}(h^c_{t+1})(x^{k}_{t+1})}{|\{x'^k_{1:t}\in \mathcal{X}^k_{1:t} :\,  \hat \mu^c_{t+1}(h^c_{t+1})(x'^k_{1:t},x^k_{t+1}) \neq 0 \}|} \\
\quad\quad\quad\text{ when }\hat \mu^c_{t+1}(h^c_{t+1})(x^k_{1:t+1}) \neq 0.
\end{array}\right.
\label{eq:mutplusupdatec_zero}
\end{align}

Then \eqref{eq:mut} at $t+1$ follows directly from the above construction.
We proceed to prove \eqref{eq:muct} at $t+1$. 
%For any $k\in\mathcal{N}$ we have
%\begin{align}
%&\mu^c_{t+1}(h^c_{t+1})(x^k_{t+1})
%=\sum_{x^k_{1:t}\in\mathcal{X}^k_{1:t}}\mu^c_{t+1}(h^c_{t+1})(x^k_{1:t+1})
%\end{align}

First consider the case when the denominator of \eqref{eq:mutupdatec_nonzero} is zero. Then,
for any $k\in\mathcal{N}$, we obtain, because of \eqref{eq:mutplusupdatec_zero},
\begin{align}
 &\mu^c_{t+1}(h^c_{t+1})(x^k_{t+1}) \nonumber\\
=&\sum_{x^k_{1:t}\in\mathcal{X}^k_{1:t}}\mu^c_{t+1}(h^c_{t+1})(x^k_{1:t+1}) \nonumber\\
=&\quad\quad\quad\sum_{\mathclap{\substack{x^k_{1:t}\in\mathcal{X}^k_{1:t}: \\ \hat \mu^c_{t+1}(h^c_{t+1})(x^k_{1:t+1}) \neq 0 }} }\quad
\frac{\gamma_{\psi,t+1}(h^c_{t+1})(x^{k}_{t+1})}{|\{x'^k_{1:t}\in \mathcal{X}^k_{1:t} :\,  \hat \mu^c_{t+1}(h^c_{t+1})(x'^k_{1:t},x^k_{t+1}) \neq 0 \}|}\nonumber\\
%=&\gamma_{\psi,t+1}(h^c_{t+1})(x^{k}_{t+1})
%\frac{|\{x^k_{1:t}\in \mathcal{X}^k_{1:t} :\,  \hat \mu^c_{t+1}(h^c_{t+1})(x^k_{1:t},x^k_{t+1}) \neq 0 \}|}{|\{x'^k_{1:t}\in \mathcal{X}^k_{1:t} :\,  \hat \mu^c_{t+1}(h^c_{t+1})(x'^k_{1:t},x^k_{t+1}) \neq 0 \}|}
%\nonumber\\
=&\gamma_{\psi,t+1}(h^c_{t+1})(x^{k}_{t+1}).
\label{eq:muct_step1}
\end{align}
%When $\hat \mu(h^c_{t+1})(x^k_{1:t+1})=0$, from the definition of signaling free beliefs $\hat\gamma(h^c_{t+1})(x^k_{t+1})= \hat \mu(h^c_{t+1})(x^k_{1:t+1})=0$. Therefore,
%\begin{align}
%\gamma_{\psi}(h^c_{t+1})(x^k_{t+1}) 
%= \hat\gamma(h^c_{t+1})(x^k_{t+1}) 
%= 0 = \mu(h^c_{t+1})(x^k_{t+1})
%\end{align}
When the denominator of \eqref{eq:mutupdatec_nonzero} is non-zero, from \eqref{eq:mutupdatec_nonzero} we get
\begin{align}
 &\mu^c_{t+1}(h^c_{t+1})(x^k_{t+1}) \nonumber\\
=&\sum_{x^k_{1:t}\in\mathcal{X}^k_{1:t}}\mu^c_{t+1}(h^c_{t+1})(x^k_{1:t+1}) \nonumber\\
=&\sum_{x^k_{1:t}\in\mathcal{X}^k_{1:t}} \frac{p^k_t(x^k_{t+1};x^k_t,a_t) \eta^k_t(x^k_{t},y^k_t,a_t,h^c_t)\mu^c_{t}(h^c_{t})(x^{k}_{1:t})}
	{\sum_{x'^k_{t} \in \mathcal{X}^k_t}\eta^k_t(x'^k_{t},y^k_t,a_t,h^c_t)  \mu^c_{t}(h^c_{t})(x'^{k}_t)} 
\nonumber\\
=&\sum_{x^k_{t}\in\mathcal{X}^k_{t}} \frac{p^k_t(x^k_{t+1};x^k_t,a_t) \eta^k_t(x^k_{t},y^k_t,a_t,h^c_t)\mu^c_{t}(h^c_{t})(x^{k}_{t})}
	{\sum_{x'^k_{t} \in \mathcal{X}^k_t}\eta^k_t(x'^k_{t},y^k_t,a_t,h^c_t)  \mu^c_{t}(h^c_{t})(x'^{k}_t)} 
\nonumber\\
\stackrel{(a)}{=}&\sum_{x^k_{t}\in\mathcal{X}^k_{t}} \frac{p^k_t(x^k_{t+1};x^k_t,a_t) \eta^k_t(x^k_{t},y^k_t,a_t,h^c_t)\gamma_{\psi,t}(h^c_{t})(x^{k}_{t})}
	{\sum_{x'^k_{t} \in \mathcal{X}^k_t}\eta^k_t(x'^k_{t},y^k_t,a_t,h^c_t)  \gamma_{\psi,t}(h^c_{t})(x'^{k}_t)} 
\nonumber\\
=&\gamma_{\psi,t+1}(h^c_{t+1})(x^{k}_{t+1}),
\label{eq:muct_step}
\end{align}
where (a) in \eqref{eq:muct_step} follows from the induction hypothesis for \eqref{eq:mut} at time $t$, and the last equality in \eqref{eq:muct_step} is true because of \eqref{eq:tilpsiS} ($\psi$ is consistent with $\lambda$).
%\\
%Then, for any $n\in\mathcal{N}$, from \eqref{eq:muct_step1} and \eqref{eq:muct_step} we obtain
%\begin{align}
%  &\mu^n_{t+1}(h^n_{t+1})(x_{t+1}) \nonumber\\
%%= &\sum_{x_{1:t} \in \mathcal{X}_{1:t}}\mu^n_{t+1}(h^n_{t+1})(x_{1:t+1}) \nonumber\\
%= &\sum_{x_{1:t}\in \mathcal{X}_{1:t}}\mathbf{1}_{\{x^n_{1:t+1}\}}(h^n_{t+1})
%\prod_{k\neq n}\mu^c_{t+1}(h^c_{t+1})(x^k_{1:t+1}) \nonumber\\
%= &\mathbf{1}_{\{x^n_{t+1}\}}(h^n_{t+1})
%\prod_{k\neq n}\mu^c_{t+1}(h^c_{t+1})(x^k_{t+1}) \nonumber\\
%= &\mathbf{1}_{\{x^n_{t+1}\}}(h^n_{t+1})\prod_{k \neq n}\gamma_{\psi,t+1}(h^c_{t+1})(x^{k}_{t+1}).
%\end{align}

Therefore, \eqref{eq:muct} is true at time $t+1$ from \eqref{eq:muct_step1} and \eqref{eq:muct_step}.

To show consistency at time $t+1$, we need to show that Bayes' rule, given by \eqref{eq:consist_positive}, is satisfied when the denominator of \eqref{eq:consist_positive} is non-zero, and \eqref{eq:consist_sigfree} holds for any histories $h^n_{t+1}$.

We first note that, the construction \eqref{eq:mutupdatec_nonzero}, \eqref{eq:mutplusupdatec_zero} and the definition of signaling-free belief $\hat \mu^c_{t+1}$ ensure that
\begin{align}
&\mu^c_{t+1}(h^c_{t+1})(x^k_{1:t+1}) = 0 \text{ if } \hat \mu^c_{t+1}(h^c_{t+1})(x^k_{1:t+1}) =0.
\label{eq:sfbeliefconsistent}
\end{align}
Therefore, \eqref{eq:consist_sigfree} follows from \eqref{eq:sfbeliefconsistent} since the signaling-free belief satisfies 
\begin{align}
&\hat\mu^n_{t+1}(h^{n}_{t+1})(x_{1:t+1}) \nonumber\\
= &\mathbf{1}_{\{x^n_{1:t+1}\}}(h^n_{t+1})\prod_{k\neq n}\hat\mu^c_{t}(h^{c}_{t+1})(x^k_{1:t+1})
\end{align}
which follows by an argument similar to that of Lemma \ref{lm:SFupdate}.

%\begin{align}
%&\mu^n_{t+1}(h^{n}_{t+1})(x_{1:t+1})= \nonumber\\
% &\frac{\mathbf{1}_{\{x^n_{1:t+1}\}}(h^n_{t+1})\mathbb{P}^g_{\mu(h^{n}_{t})}(x_{1:t+1},c_{t+1},a_t,y_t|h^{n}_t)}
%{\sum_{x'_{1:t}\in\mathcal{X}_{1:t+1}, x'^n_{t+1}=x^n_{t+1}}\mathbb{P}^g_{\mu(h^{n}_{t})}(x'_{1:t+1},c_{t+1},a_t,y_t|h^{n}_t)}
%\label{eq:BRnt}
%\end{align}
%is satisfied for any histories $h^c_{t+1}$ and $h^n_{t+1}, n\in \mathcal{N}$ when the above denominator is non-zero.
%\\
Now consider \eqref{eq:consist_positive} at $t+1$ when the denominator is non-zero.
From \eqref{eq:mutupdate1}-\eqref{eq:mutupdatec_nonzero} the left hand side of \eqref{eq:consist_positive} equals to
\begin{align}
&\mu^n_{t+1}(h^n_{t+1})(x_{1:t+1}) 
=  \mathbf{1}_{\{x^n_{1:t+1}\}}(h^n_{t+1}) \nonumber\\
&\quad \prod_{k\neq n}
\frac{p^k_t(x^k_{t+1};x^k_t,a_t) \eta^k_t(x^k_{t},y^k_t,a_t,h^c_t)\mu^c_{t}(h^c_{t})(x^{k}_{1:t})}
	{\sum_{x'^k_{t} \in \mathcal{X}^k_t}\eta^k_t(x'^k_{t},y^k_t,a_t,h^c_t)  \mu^c_t(h^c_{t})(x'^{k}_t)} .
%\nonumber\\
%= & \mathbf{1}_{\{x^n_{1:t+1}\}}(h^n_{t+1}) \prod_{k\neq n}
%\mu^c_t(h^c_{t})(x^{k}_{1:t})\frac{p^k_t(x^k_{t+1};x^k_t,a_t) q^k_t(y^k_{t};x^k_t,a_t)\lambda^{k}_t(x^k_t,c_t,\pi_{t},\hat \pi_t)(a^k_t)}
%	{\sum_{x'^k_{t} \in \mathcal{X}^k_t}q^k_t(y^k_{t};x'^k_t,a_t)\lambda^{k}_t(x'^k_t,c_t,\pi_{t},\hat \pi_t)(a^k_t)  \mu^c_t(h^c_{t})(x'^{k}_t)} . 
\label{eq:BRntL}
\end{align}
%when the above denominator is non-zero.
%the last equality in \eqref{eq:BRntL} follows from the induction hypothesis for \eqref{eq:muct} at time $t$.

On the other hand, the numerator of the right hand side of \eqref{eq:consist_positive} is equal to
\begin{align}
 &\mathbb{P}^{g_t}_{\mu}(x_{1:t+1},y_t,a_t|h^n_t,a^n_t) \nonumber\\
= &\mu^n_t(h^{n}_{t})(x_{1:t})\prod_{k\in\mathcal{N}} p^k_t(x^k_{t+1};x^k_t,a_t)q^k_t(y^k_{t};x^k_t,a_t) \nonumber\\
   &\quad\prod_{k\neq n}\lambda_t^{k}(x^k_t,c_t,\gamma_{\psi,t}(h^c_t),\hat \gamma_t(h^c_t))(a^k_t) \nonumber\\
= &\mathbf{1}_{\{x^n_{1:t}\}}(h^n_{t})\prod_{k\in\mathcal{N}} p^k_t(x^k_{t+1};x^k_t,a_t)q^k_t(y^k_{t};x^k_t,a_t) \nonumber\\
   &\quad\prod_{k\neq n} \mu^c_t(h^{c}_{t})(x^k_{1:t})\prod_{k\neq n}\lambda_t^{k}(x^k_t,c_t,\gamma_{\psi,t}(h^c_t),\hat \gamma_t(h^c_t))(a^k_t) 
\nonumber\\
= &\mathbf{1}_{\{x^n_{1:t}\}}(h^n_{t}) p^n_t(x^n_{t+1};x^n_t,a_t)q^n_t(y^n_{t};x^n_t,a_t) \nonumber\\
   &\quad\prod_{k\neq n}p^k_t(x^k_{t+1};x^k_t,a_t) \eta^k_t(x^k_{t},y^k_t,a_t,h^c_t)\mu^c_{t}(h^c_{t})(x^{k}_{1:t})
\label{eq:consistencyinproof}
\end{align}
The first equality in \eqref{eq:consistencyinproof} follows from \eqref{eq:conditionalupdate} and \eqref{eq:getaS}. The second equality in \eqref{eq:consistencyinproof} follows from the induction hypothesis for \eqref{eq:mut}. The last equality in \eqref{eq:consistencyinproof} follows from \eqref{eq:etainproof}.

%\begin{align}
%&\mathbf{1}_{\{x^n_{1:t+1}\}}(h^n_{t+1})\frac{\mu^n_t(h^{n}_{t})(x^{-n}_{1:t})\prod_{k\neq n}p^k_t(x^k_{t+1};x^k_t,a_t) q^k_t(y^k_{t};x^k_t,a_t)g^{k}_t(x^k_{1:t},h^c_t)(a^k_t)}
%	{\sum_{x'^{-n}_{1:t}\in\mathcal{X}^{-n}_{1:t}}\mu^n_t(h^{n}_{t})(x'^{-n}_{1:t})\prod_{k\neq n} q^k_t(y^k_{t};x'^k_t,a_t)g^{k}_t(x'^k_{1:t},h^c_t)(a^k_t)} 
%\nonumber\\
%= &\mathbf{1}_{\{x^n_{1:t+1}\}}(h^n_{t+1})\frac{\prod_{k\neq n}\mu^c_t(h^{c}_{t})(x^k_{1:t})p^k_t(x^k_{t+1};x^k_t,a_t) q^k_t(y^k_{t};x^k_t,a_t)g^{k}_t(x^k_{1:t},h^c_t)(a^k_t)}
%	{\sum_{x'^{-n}_{1:t}\in\mathcal{X}^{-n}_{1:t}}\prod_{k\neq n} \mu^c_t(h^{c}_{t})(x'^k_{1:t})q^k_t(y^k_{t};x'^k_t,a_t)g^{k}_t(x'^k_{1:t},h^c_t)(a^k_t)} 
%\nonumber\\
%= &\mathbf{1}_{\{x^n_{1:t+1}\}}(h^n_{t+1})\frac{\prod_{k\neq n}\mu^c_t(h^{c}_{t})(x^k_{1:t})p^k_t(x^k_{t+1};x^k_t,a_t) q^k_t(y^k_{t};x^k_t,a_t)g^{k}_t(x^k_{1:t},h^c_t)(a^k_t)}
%	{\prod_{k\neq n} \sum_{x'^k_{1:t}\in\mathcal{X}^{k}_{1:t}}\mu^c_t(h^{c}_{t})(x'^k_{1:t})q^k_t(y^k_{t};x'^k_t,a_t)g^{k}_t(x'^k_{1:t},h^c_t)(a^k_t)} 
%\nonumber\\
%=& \text{ the right hand side of \eqref{eq:BRntL}}.
%\label{eq:consistencyinproof}
%\end{align}

Substituting \eqref{eq:consistencyinproof} back into both the numerator and the denominator in the right hand side of \eqref{eq:consist_positive}, we obtain \eqref{eq:BRntL}. Therefore, \eqref{eq:consist_positive} is satisfied for any history $h^n_{t+1} \in \mathcal{H}^n_{t+1}$ for any $n\in\mathcal{N}$ when the denominator of \eqref{eq:consist_positive} is non-zero, hence, $(g,\mu)$ is consistence before time $t+1$. This completes the induction step and the proof of the lemma.

%because of the dynamics of the Markov chains and the specification of $g$, the numerator in the right hand side of \eqref{eq:BRnt} is equal to
%\begin{align}
%&\mathbf{1}_{\{x^n_{1:t+1}\}}(h^n_{t+1})\mathbb{P}^g_{\mu(h^{n}_{t})}(x_{1:t+1},c_{t+1},a_t,y_t|h^{n}_t)\nonumber\\
%=&\mathbf{1}_{\{x^n_{1:t+1}\}}(h^n_{t+1}) \mathbb{P}(c_{t+1}|c_t,a_t)
%\prod_{k\neq n}\mu(h^c_t)(x^k_{1:t})
%\mathbb{P}(x^k_{t+1}|x^k_t,a_t)
%\mathbb{P}(y^k_{t}|x^k_t,a_t)\lambda^{k}_t(x^k_t,c_t,\pi_{t},\hat \pi_t)(a^k_t)
%\label{eq:BRntnum}
%\end{align}
%Substituting \eqref{eq:BRntnum} back into the right hand side of \eqref{eq:BRnt}, we find that
%the right hand side of \eqref{eq:BRnt} is equal to the right hand side of \eqref{eq:BRntL}. Therefore, \eqref{eq:BRnt} is satisfied for any history $h^n_{t+1} \in \mathcal{H}_{t+1}$, hence, $(g,\mu)$ is consistence.

\end{proof}

\begin{proof}[Proof of Lemma \ref{lm:condindep}]
If agent $n$ uses an arbitrary strategy $g'^n$, following the same construction \eqref{eq:mutupdate1}-\eqref{eq:mutupdatec_nonzero} in Lemma \ref{lm:befupdate}, we can obtain a belief system $\mu'$ from $g' := (g'^n,g^{-n})$ and $\psi$ such that
\begin{align}
\mu'^n_t(h^n_t)(x_{1:t}) = \mathbf{1}_{\{x^n_{1:t}\}}(h^{n}_t) \prod_{k \neq n} \mu'^c_t(h^c_t)(x^{k}_{1:t}).
\end{align}
Since $\mu'^c_t(h^c_t)(x^{k}_{1:t})$ defined by \eqref{eq:mutupdatec_nonzero} and \eqref{eq:mutplusupdatec_zero} depends only on the strategies $g'^{-n}=g^{-n}$ of all agents other than $n$, we have for all $h^c_t\in\mathcal{H}^c_t$
\begin{align}
\mu'^c_t(h^c_t)(x^{k}_{1:t}) = \mu^c_t(h^c_t)(x^{k}_{1:t}).
% \text{ for all }x^{k}_{1:t} \in \mathcal{X}^{k}_{1:t}.
\end{align}
Therefore, for any history $H^n_t\in \mathcal{H}^n_t$
\begin{align}
 \mu'^n_t(h^n_t)(x_{1:t}) = \mathbf{1}_{\{x^n_{1:t}\}}(h^{n}_t) \prod_{k \neq n} \mu^c_t(h^c_t)(x^{k}_{1:t}).
\end{align}
The same argument for the proof of consistency in Lemma \ref{lm:befupdate} shows that $\mu'$ is consistent with $g'=( g'^n,g^{-n})$. Therefore,
when $\mathbb{P}^{g'^n,g^{-n}}(h^n_t)>0$, from Bayes' rule we have
\begin{align}
\mathbb{P}^{g'^n,g^{-n}}(x_{1:t}|h^n_t) =
&\mathbb{P}^{g'^n,g^{-n}}_{\mu'}(x_{1:t}|h^n_t) =
\mu'^n_t(h^n_t)(x_{1:t})  \nonumber\\
= &\mathbf{1}_{\{x^n_{1:t}\}}(h^{n}_t) \prod_{k \neq n} \mu^c_t(h^c_t)(x^{k}_{1:t}).
\end{align}
\end{proof}

\begin{proof}[Proof of Lemma \ref{lm:Commonclose}]
%To simply the notation, we use $\Pi_t$ to denote the \CI\ belief $\Pi^{\gamma_\psi}_t$ constructed from $\psi$, and $B_t = (C_t,\Pi_t,\hat\Pi_t)$.
%To simply the notation, let $B_t = (C_t,\Pi^{\gamma_\psi}_t,\hat\Pi_t)= (C_t,\gamma_{\psi,t}(H^c_t),\hat\gamma_t(H^c_t))$.
To simply the notation, we use $\Pi_t$ to denote $\Pi^{\gamma_\psi}_t$ and $B_t = (C_t,\Pi_t,\hat\Pi_t)$.

Let $(g,\mu)=f(\lambda,\psi)$ as in Lemma \ref{lm:befupdate}.
Suppose every agent $k\neq n$ uses the strategy $g^{k}$ along with the belief system $\mu$.

Below, we show that agent $n$'s best response problem \eqref{eq:seqrational} is a Markov Decision Process (MDP) with state process $\{(X^n_t,B_t), t\in\mathcal{T}\}$ and action process $\{A^n_t,t\in\mathcal{T}\}$.

Since the strategies $g^{-n}$ of all other agents are fixed, when agent $n$ selects an action $a^n_t\in\mathcal{A}^n_t$ at time $t\in\mathcal{T}$, agent $n$'s expected instantaneous utility at $h^n_t \in \mathcal{H}^n_t$ under $\mu$ is given by
\begin{align}
&\mathbb{E}^{g^{-n}}_{\mu}\left[
\phi^{n}_t(C_t,X_t,A_t)|h^n_t,a^n_t\right].
 \label{eq:inst_utilityforn}
\end{align}
Since $A^k_t, k\neq n$ satisfies \eqref{eq:getaS}, 
the distribution of $A^k_t$ only depends on $X^k_t$ and $B_t$.
%Note that $\Pi_{t+1}= \psi_t(\Pi_{t},\hat \Pi_t,A_t,Y_t)$ from \eqref{eq:psiB}.
Therefore, the distribution of $A^{-n}_t$ only depends on $X^{-n}_t$ and $B_t$. 
Then, for any realization $x^{-n}_t \in \mathcal{X}^{-n}_t$, $h^n_t=(x^{n}_{1:t},h^c_t) \in\mathcal{H}^{n}_t$ and $a^n_t \in \mathcal{A}^n_t$,
\begin{align}
&\mathbb{E}^{g^{-n}}_{\mu}\left[
\phi^{n}_t(C_t,X_t,A_t)| x^{-n}_t,a^n_t, h^n_t\right]
\nonumber\\
=&\mathbb{E}^{g^{-n}}_{\mu}\!\!\left[
\phi^{n}_t(c_t,x_t,(a^n_t,A^{-n}_t))| x^{-n}_t,x^{n}_{1:t},a^n_t,h^c_t,b_t\right]
\nonumber\\
=&\mathbb{E}^{g^{-n}_t}\left[
\phi^{n}_t(c_t,x_t,(a^n_t,A^{-n}_t))| x^{-n}_t,b_t\right]
\nonumber\\
=:&\bar\phi^{n}_t(x_t,a^n_t,b_t,g^{-n}_t);
\label{eq:tildephiS}
\end{align}
the first equality in \eqref{eq:tildephiS} holds because given $\psi$, $B_t = (C_t,\gamma_{\psi,t}(H^c_t),\hat\gamma_t(H^c_t))$ is a function of $H^c_t$;
the second equality in \eqref{eq:tildephiS} is true because the distribution of $A^{-n}_t$ depends only on $X^{-n}_t$, $B_t$ and the strategy $g^{-n}_t$.
From \eqref{eq:tildephiS}, agent $n$'s instantaneous utility \eqref{eq:inst_utilityforn} can be written as
\begin{align}
&\mathbb{E}^{g^{-n}}_{\mu}\left[
\phi^{n}_t(C_t,X_t,A_t)|h^n_t,a^n_t\right]
\nonumber\\
= &\mathbb{E}^{g^{-n}}_{\mu}\left[\mathbb{E}^{g^{-n}}_{\mu}\left[
\phi^{n}_t(C_t,X_t,A_t)| X^{-n}_t,a^n_t, h^n_t\right]|h^n_t,a^n_t\right]
\nonumber\\
= &\mathbb{E}^{g^{-n}}_{\mu}\left[\bar\phi^{n}_t((x^n_t,X^{-n}_t),a^n_t,b_t,g^{-n}_t) |h^c_t,x^n_{1:t},a^n_t\right]
\nonumber\\
\stackrel{(a)}{=}&\mathbb{E}_{\mu}\left[\bar\phi^{n}_t((X^{-n}_t,x^n_t),a^n_t,b_{t},g^{-n}_t) |h^c_t\right]
\nonumber\\
\stackrel{(b)}{=}&\mathbb{E}_{\pi_t}\left[\bar\phi^{n}_t((X^{-n}_t,x^n_t),a^n_t,b_{t},g^{-n}_t)\right]
\nonumber\\
=: & \tilde{\phi}^{n}_t (x^n_t,b_t,a^n_t,g^{-n}_t).
 \label{eq:utilityforn}
\end{align}
Equation (a) is true because, from Lemma \ref{lm:condindep},
$X^{-n}_t$ and $X^{n}_t$ are independent conditional on $h^c_t$.
Equation (b) follows from the fact that $\pi_t$ is the distribution of $X_t$ conditional on $h^c_t$ under $\mu$, which is established by \eqref{eq:muct} in Lemma \ref{lm:befupdate}.

Next, we show that the process $\{(X^n_t,B_t), t\in\mathcal{T}\}$ is a controlled Markov chain with respect to the action process $\{A^n_t,t\in\mathcal{T}\}$ for agent $n$.

From the system dynamics and the belief evolution \eqref{eq:psiB}, we know that $(X^n_{t+1},B_{t+1})$ is a function of 
$\{X^n_t,Y^{-n}_t,A^n_t,A^{-n}_t,B_t,W_t\}$ where $W_t$ denotes all the noises at time $t$.
Furthermore, the distribution of $(Y^k_t,A^k_t)$ depends only on $\{X^k_t,B_t,W_t,g^k_t\}$ for any $k\neq n$.
Therefore,
\begin{align}
(X^n_{t+1},B_{t+1}) =
\tilde f_t(X^n_t,X^{-n}_t,A^n_t,B_t,W_t,g^{-n}_t).
\label{eq:distridepend}
\end{align}
Suppose agent $n$ uses an arbitrary strategy $\tilde g^n$.
Then, for any realizations  $x^n_{t+1}\in \mathcal{X}^n_{t+1}$, $b_{t+1}=(c_{t+1},\pi_{t+1},\hat\pi_{t+1}) \in \mathcal{B}_{t+1}$, $h^n_t=(x^{n}_{1:t},h^c_t) \in\mathcal{H}^{n}_t$ and $a^n_t \in \mathcal{A}^n_t$, we obtain
\begin{align}
 &\mathbb{P}^{\tilde g^n,g^{-n}}_{\mu}(x^n_{t+1},b_{t+1}|h^n_t,a^n_t ) \nonumber\\
= & \sum_{\mathclap{x^{-n}_t \in \mathcal{X}^{-n}_t}}\mathbb{P}^{\tilde g^n,g^{-n}}_{\mu}(x^n_{t+1},b_{t+1}|x^{-n}_t,h^n_t,a^n_t )
\mathbb{P}^{\tilde g^n,g^{-n}}_{\mu}(x^{-n}_t|h^n_t,a^n_t )
 \nonumber\\
 = & \sum_{\mathclap{x^{-n}_t \in \mathcal{X}^{-n}_t}}\mathbb{P}^{\tilde g^n,g^{-n}}_{\mu}(x^n_{t+1},b_{t+1}|x^{-n}_t,h^n_t,a^n_t )
\pi_t(x^{-n}_t)
 \nonumber\\
 = & \sum_{\mathclap{x^{-n}_t \in \mathcal{X}^{-n}_t}}\mathbb{P}^{g^{-n}_t}(x^n_{t+1},b_{t+1}|x^n_t,x^{-n}_t,b_t,a^n_t )
\pi_t(x^{-n}_t)
 \nonumber\\
 = & \mathbb{P}^{g^{-n}}_{\mu}(x^n_{t+1},b_{t+1}|x^n_t,b_t,a^n_t).
\label{eq:conMC}
\end{align}
The second equality in \eqref{eq:conMC} follows from Lemma \ref{lm:condindep} and \eqref{eq:muct} in Lemma \ref{lm:befupdate}.
The third equality in \eqref{eq:conMC} follows from \eqref{eq:distridepend}.
The last equality follows from the same arguments as the first through third equalities.

Equation \eqref{eq:conMC} shows that the process $\{(X^n_t,B_t), t\in\mathcal{T}\}$ is a controlled Markov Chain with respect to 
the action process $\{A^n_t,t\in\mathcal{T}\}$ for agent $n$.
This process along with the instantaneous utility  \eqref{eq:utilityforn} define a MDP.
From the theory of MDP (see \cite[Chap. 6]{kumar1986stochastic}), there is an optimal strategy of agent $n$ that is of the form
\begin{align}
\lambda'^{n}_t(x^n_t,b_t) = \lambda'^{n}_t(x^n_t,(c_t,\gamma_{\psi,t}(h^c_t),\hat\gamma_t(h^c_t)))
%g'^n_t(H^n_t) = \lambda'^{n}_t(X^n_t,B_t)
%=\lambda'^{n}(X^n_t,C_t,\gamma_{\psi,t}(H^c_t),\hat\gamma_t(H^c_t))
\end{align}
for all $h^n_t=(x^n_{1:t},h^c_t)\in\mathcal{H}^n_t$ for all $t\in\mathcal{T}$. 
%for some \CI\ strategy $\lambda'^{n}_t$, and $g'^n_t$ is a best response for agent $n$ at every $h^n_t\in\mathcal{H}^n_t$ for all $t\in\mathcal{T}$. 
This completes the proof of Lemma \ref{lm:Commonclose}.

\end{proof}

\begin{proof}[Proof of Theorem \ref{thm:commonDP}]
Suppose $(\lambda^*,\psi^*)$ solves the dynamic program defined by \eqref{eq:DP:end}-\eqref{eq:DP:valueupdate}. Let $V^{n}_t,n\in\mathcal{N},t\in\mathcal{T},$ denote the value functions computed by \eqref{eq:DP:end} and \eqref{eq:DP:valueupdate} from $(\lambda^*,\psi^*)$.
Then $\psi^*$ is consistent with $\lambda^*$ from \eqref{eq:DP:consistent}. 
%the functions $\left\{V_{t}^{*n}(\cdot), n\in\mathcal{N},t \in\mathcal{T}\right\} = V(\lambda^*,\psi^*)$. 

Let $(g^*,\mu^*)=f(\lambda^*,\psi^*)$ defined by Lemma \ref{lm:befupdate}.
Then $\mu^*$ is consistent with $g^*$ because of Lemma \ref{lm:befupdate}. Furthermore, for all $n\in\mathcal{N}, t\in\mathcal{T},$
 $V^{n}_t(x^n_t,b_t)$ (where $b_t=(c_t,\gamma_{\psi^*,t}(h^c_t), \hat\gamma_t(h^c_t))$) is agent $n$'s expected continuation utility from time $t$ on under $\mu^*$ at $h^n_t=(x^n_{1:t},h^c_t)$ when agent $n$ uses $g^{*n}$ and all other agents use $g^{*-n}$.

If every agent $k \neq n$ uses the strategy $g^{*k}$, 
from Lemma \ref{lm:Commonclose} we know that there is a best response $g'^{n}$, under the belief system $\mu^*$, of agent $n$ such that for all $t\in\mathcal{T}$
\begin{align}
&g'^n_t(h^n_t) = \lambda'^n_t(x^n_t,b_t)
\label{eq:gtildeS}
\end{align}
for some \CI\ strategy $ \lambda'^n_t$ for all $h^n_t=(h^c_t,x^n_{1:t})$.
Define a \CI\ strategy profile $\lambda' := (\lambda'^{n}, \lambda^{*-n})$.
% and the corresponding strategy profile $g' := (g'^{n}, g^{*-n})$ (by Lemma \ref{lm:befupdate}). 

Let $V'^n_{t}, n\in\mathcal{N},t \in\mathcal{T},$ be the functions generated by \eqref{eq:DP:end} and \eqref{eq:DP:valueupdate} from $(\lambda',\psi^*)$.
Then $V'^{n}_t(x^n_t,b_t)$ (where $b_t=(c_t,\gamma_{\psi^*,t}(h^c_t), \hat\gamma_t(h^c_t))$) is agent $n$'s expected continuation utility from time $t$ on under $\mu^*$ at $h^n_t=(x^n_{1:t},h^c_t)$ when agent $n$ uses $g'^{n}$ and all other agents use $g^{*-n}$.
Since $g'^{n}$ is a best response, for all $n\in\mathcal{N},t\in\mathcal{T}$ and $h^n_t=(x^n_{1:t},h^c_t)\in\mathcal{H}^n_t$ we must have
\begin{align}
V'^{n}_t( x^n_t,b_t) \geq  V^{n}_t( x^n_t,b_t).
\label{eq:valuecomp1}
\end{align}

On the other hand, $V'^{n}_t( x^n_t,b_t)$ is player $n$'s expected utility in stage game $G_{t}(V_{t+1},\psi^*_t,b_t)$
%, where $b_t=(c_t,\gamma_{\psi^*,t}(h^c_t), \hat\gamma_t(h^c_t))$, 
when player $n$ uses $\lambda'^{n}|_{b_t}$, and other players use $\lambda^{*-n}|_{b_t}$.
%when player $n$ uses $\lambda'^{n}|_{b_t}$ that corresponds to $\lambda'^{n}$, and other players use $\tilde\lambda^{*-n}$ that corresponds to $\lambda^{*-n}$.
However, from \eqref{eq:DP:BNE}, $V^{n}_t(x^n_t,b_t)$ is player $n$'s maximum expected utility in stage game $G_{t}(V_{t+1},\psi^*_t,b_t)$ when other players use $\lambda^{*-n}|_{b_t}$ because the strategy $\lambda^{*n}_t|_{b_t}$ is a best response for player $n$ in the stage game. This means that for all $n\in\mathcal{N},t\in\mathcal{T}$ and
$b_t\in\mathcal{B}_t$
% $h^n_t=(x^n_{1:t},h^c_t)\in\mathcal{H}^n_t$, $b_t=(c_t,\gamma_{\psi^*,t}(h^c_t), \hat\gamma_t(h^c_t))$
\begin{align}
V^{n}_t( x^n_t,b_t) \geq  V'^{n}_t( x^n_t,b_t).
\label{eq:valuecomp2}
\end{align}
%$\lambda_t^{*n}(x^n_t,c_t,\gamma(h^c_t),\hat\gamma(h^c_t))$ is one of the optimal solutions of \eqref{eq:optimalS} for all $(x^n_t,c_t,h^c_t)$ for all $t$, we can show by induction that all $(x^n_t,c_t,h^c_t)$, for all $t$,
%\begin{align}
%V^{*n}_t(x^n_t,c_t,\gamma(h^c_t),\hat\gamma(h^c_t)) = \tilde V^{n}_t( x^n_t,c_t,\gamma(h^c_t),\hat\gamma(h^c_t)).
%\label{eq:valuecomp1}
%\end{align}
Combining \eqref{eq:valuecomp1} and \eqref{eq:valuecomp1} we get
\begin{align}
V^{n}_t( x^n_t,b_t) =  V'^{n}_t( x^n_t,b_t).
\label{eq:valucomp}
\end{align}
Equation \eqref{eq:valucomp} implies that, at any time $t$, the strategy $g^{*n}_{t:T}$ gives agent $n$ the maximum expected continuation utility from time $t$ on under $\mu^*$. This complete the proof that $(g^*,\mu^*)$ is a PBE.
As a result, the pair $(\lambda^*,\psi^*)$ forms a \CI-PBE of the dynamic game described in Section \ref{sec:model}.

%Let $\left\{\tilde V_{t}^{n}(\cdot), n\in\mathcal{N},t \in\mathcal{T}\right\} := V(\tilde \lambda, \psi^*)$, then
%$\tilde V^{n}_t(x^n_t,c_t,\gamma(h^c_t),\hat\gamma(h^c_t))$ is the expected continuation revenue from time $t$ on, under $\mu^*$, for agent $n$ at $h^n_t$ under the strategy profile $\tilde g$.
%Since $\tilde g^{n}$ is a best response, $\tilde V^{n}_t( x^n_t,c_t,\gamma(h^c_t),\hat\gamma(h^c_t))$ gives agent $n$ the maximum expected continuation revenue from time $t$ on under $\mu^*$.
%However, since $\lambda_t^{*n}(x^n_t,c_t,\gamma(h^c_t),\hat\gamma(h^c_t))$ is one of the optimal solutions of \eqref{eq:optimalS} for all $(x^n_t,c_t,h^c_t)$ for all $t$, we can show by induction that all $(x^n_t,c_t,h^c_t)$, for all $t$,
%\begin{align}
%V^{*n}_t(x^n_t,c_t,\gamma(h^c_t),\hat\gamma(h^c_t)) = \tilde V^{n}_t( x^n_t,c_t,\gamma(h^c_t),\hat\gamma(h^c_t)).
%\label{eq:valuecomp1}
%\end{align}
%Therefore, \eqref{eq:valuecomp1} implies that, at any time $t$, $g^{*n}_{t:T}$ gives agent $n$ the maximum expected continuation revenue from time $t$ on under $\mu^*$. This complete the proof that $(g^*,\mu^*)$ is a PBE.
%As a result, the pair $(\lambda^*,\psi^*)$ forms a \CI\ equilibrium of the dynamic game described in Section \ref{sec:model}.
\end{proof}

\section{}
\label{app:existence}

In order to prove Theorem \ref{thm:publiccommon}, we first prove the following lemma.
\begin{lemma}
\label{lm:SFbefACind}
In \textbf{Game M}
\begin{align}
\hat\pi^n_{t+1} = \hat\psi^n_t(y^n_t,\hat\pi_{t}). 
\end{align}
\end{lemma}
%\begin{proof}
%See Appendix \ref{app:existence}.
%\end{proof}}

\begin{proof}[Proof of Lemma \ref{lm:SFbefACind}]
From Lemma \ref{lm:SFupdate}
\begin{align}
\hspace{-0.5em}\hat\pi^n_{t+1}= &\hat\psi^n_t(y^n_t,a_t,\hat \pi_{t})(x^n_{t+1})  \nonumber\\
=&\frac{\sum_{x^n_{t} \in \mathcal{X}^n_t}p^n_t(x^n_{t+1};x^n_t,a_t) q^n_t(y^n_{t};x^n_t,a_t) \hat\pi^{n}_{t}(x^n_{t})}
	{\sum_{x'^n_{t} \in \mathcal{X}^n_t}q^n_t(y^n_{t};x'^n_t,a_t)\hat\pi^{n}_{t}(x'^n_{t})} .
\end{align}
Since $p^n_t(x^n_{t+1};x^n_t,a_t) = p^n_t(x^n_{t+1};x^n_t)$ and $q^n_t(y^n_{t};x^n_t,a_t)=q^n_t(y^n_{t};x^n_t)$ in \textbf{Game M}, the assertion of the lemma holds.
\end{proof}

Lemma \ref{lm:SFbefACind} shows that in \textbf{Game M} the signaling-free beliefs do not depend on the actions.

We now prove Theorem \ref{thm:publiccommon}.

\begin{proof}[Proof of Theorem \ref{thm:publiccommon}]
%First note that, $\hat\pi^n_{t+1} = \hat\psi^n_t(\hat\pi_{t},y^n_t)$ from Lemma \ref{lm:SFbefACind}.
Consider a \CI\ update rule $\psi^*$ given by
\begin{align}
&\psi^{n*}_t(y^n_t,a_t,b_t) = \hat\psi^n_t(y^n_t,\hat\pi_{t}). 
\label{psistarS}
\end{align}
Based on $\psi^*$ defined by \eqref{psistarS}, we solve the dynamic program defined by \eqref{eq:DP:end}-\eqref{eq:DP:valueupdate} to get a \CI\ strategy profile $\lambda^*$ and show that $(\lambda^*,\psi^*)$ forms a \CI-PBE for \textbf{Game M}.
Note that under the update rule $\psi^*$ given by \eqref{psistarS}, we have for any $n$ and $t$
\begin{align}
\Pi^n_{t} = \hat\Pi^n_{t}
\end{align}
Therefore, in the following we will replace $\Pi^n_{t}$ by $\hat\Pi^n_{t}$ and drop $\Pi^n_{t}$ if both $\Pi^n_{t}$ and $\hat\Pi^n_{t}$ are present.

The dynamic program for \textbf{Game M} can be solved by induction. We prove the following claim:

At any time $t$, there exists a \CI\ strategy $\lambda^*_{t}$ that satisfies \eqref{eq:DP:BNE}, and the value functions $V^n_{t}, n\in \mathcal{N},$ generated by \eqref{eq:DP:end} and \eqref{eq:DP:valueupdate} from $(\lambda^*_{t:T},\psi^*_{t:T})$ satisfy
\begin{align}
V^n_t(x^n_t,b_t) = \tilde U^n_t(c_t,\hat\pi_t) + \tilde V^n_t(x^n_t,c_{t_{k}+1},\hat\pi_t)
\label{eq:valuedecomp}
\end{align}
for some functions $\tilde U^n_t(c_t,\hat\pi_t)$ and $\tilde V^n_t(x^n_t,c_{t_{k}+1},\hat\pi_t)$ when $t_k+1\leq t\leq t_{k+1}$ for some $t_k\in\overline{\mathcal{T}}$. 

The above claim holds at $t=T+1$ since $V^n_{T+1}=0, n\in\mathcal{N}$.
 
Suppose the claim is true at $t+1$. 

At time $ t_k+1 \leq t < t_{k+1}$ for some $t_k \in \overline{\mathcal{T}}$, $X^n_{t+1}=X^n_t$, $Y_t=$ empty, and $C_{t+1}=(C_{t},A_t)=(C_{t_k+1},A_{t_k+1:t})$. Then because of \eqref{eq:G2U}, \eqref{psistarS} and the induction hypothesis for \eqref{eq:valuedecomp}, player $n$'s utility in stage game $G_t(V_{t+1},\psi^*_t,b_t)$ is equal to
\begin{align}
&U^n_{G_t(V_{t+1},\psi^*_t,b_t)} \nonumber\\
=&\phi^{n}_t(c_t,X^{-n}_t,A_t) +V^{n}_{t+1}(X^n_{t+1},C_{t+1},\hat\psi_t(\hat\pi_{t},Y_t))
\nonumber\\
=&\phi^{n}_t(c_t,X^{-n}_t,A_t) +\tilde U^n_{t+1}((c_t,A_t),\hat\psi_t(\hat\pi_{t})) \nonumber\\
 & + \tilde V^n_{t+1}(X^n_{t},c_{t_k+1},\hat\psi_t(\hat\pi_{t}))
\label{eq:utilityinpf_CA1}
\end{align}
for any $b_t\in\mathcal{B}_t$ and $n\in\mathcal{N}$.
Define
%when $t_k+1 \leq t < t_{k+1}$,
\begin{align}
&\tilde \phi^n_t(X^{-n}_t,A_t,b_t)\nonumber\\
:= &\phi^{n}_t(c_t,X^{-n}_t,A_t) +\tilde U^n_{t+1}((c_t,A_t),\hat\psi_t(\hat\pi_{t})), 
\label{eq:utilityinpf_CA1_1}
\\
&\tilde V^n_t(X^n_t,c_{t_k+1},\hat\pi_t):= \tilde V^n_{t+1}(X^n_{t},c_{t_k+1},\hat\psi_t(\hat\pi_{t})).
\label{eq:utilityinpf_CA1_2}
\end{align}

At $t = t_{k}$ for some $t_{k} \in \overline{\mathcal{T}}$,  $X^n_{t_k+1}=f_{t_k}^n(X_{t_k}^n,W_{t_k}^{n,X})$, $Y^n_{t_k}=h_{t_k}^n(X_{t_k}^n,W_{t_k}^{n,Y}) $ and $C_{t_k+1}=f_{t_k}^c(C_{t_{k-1}+1},W_{t_k}^{C})$. Then because of \eqref{eq:G2U}, \eqref{psistarS} and the induction hypothesis for \eqref{eq:valuedecomp}, player $n$'s utility in stage game $G_{t_k}(V_{{t_k}+1},\psi^*_{t_k},b_{t_k})$ is equal to
\begin{align}
&U^n_{G_{t_k}(V_{{t_k}+1},\psi^*_{t_k},b_{t_k})} \nonumber\\
=&\phi^{n}_{t_k}(c_{t_k},X^{-n}_{t_k},A_{t_k}) +V^{n}_{{t_k}+1}(X^n_{{t_k}+1},C_{{t_k}+1},\hat\psi_t(\hat\pi_{t_k},Y_{t_k}))
\nonumber\\
=&\phi^{n}_{t_k}(c_{t_k},X^{-n}_{t_k},A_{t_k})\nonumber\\
+&\tilde U^n_{t+1}(f_{t_k}^c(c_{t_{k-1}+1},W_{t_k}^{C}),\hat\psi_t(\hat\pi_{t},h_{t_k}(X_{t_k},W_{t_k}^{Y}))) 
\nonumber\\
+& \tilde V^n_{t+1}(f_{t_k}^n(X_{t_k}^n,W_{t_k}^{n,X}),f_{t_k}^c(c_{t_{k-1}+1},W_{t_k}^{C}),\hat\psi_t(\hat\pi_{t}))
\label{eq:utilityinpf_CA2}
\end{align}
for any $b_t\in\mathcal{B}_{t_k}$ and $n\in\mathcal{N}$.
Define
\begin{align}
&\tilde \phi^n_{t_k}(c_{t_k},X^{-n}_{t_k},A_{t_k},\hat\pi_{t_k}):=\phi^{n}_{t_k}(c_{t_k},X^{-n}_{t_k},A_{t_k}), 
\label{eq:utilityinpf_CA2_1}
\\
&\tilde V^n_{t_k}(X^n_{t_k},c_{t_{k-1}+1},\hat\pi_{t_k})\nonumber\\
:=& 
\mathbb{E}_{\hat\pi_{t_k}}\left[
\tilde U^n_{t+1}(f_{t_k}^c(c_{t_{k-1}+1},W_{t_k}^{C}),\hat\psi_t(\hat\pi_{t},h_{t_k}(X_{t_k},W_{t_k}^{Y}))) \right.
\nonumber\\
+& \left.\tilde V^n_{t+1}(f_{t_k}^n(X_{t_k}^n,W_{t_k}^{n,X}),f_{t_k}^c(c_{t_{k-1}+1},W_{t_k}^{C}),\hat\psi_t(\hat\pi_{t_k}))| X_{t_k} \right].
\label{eq:utilityinpf_CA2_2}
\end{align}

Therefore, for any $t$, because of \eqref{eq:utilityinpf_CA1}-\eqref{eq:utilityinpf_CA2_2} player $n$'s expected utility conditional on $(X_t,A_t)$ in stage game $G_t(V_{t+1},\psi^*_t,b_t)$, for $b_t=(c_t,\hat\pi_t)\in\mathcal{B}_t$, is equal to
%From the above analysis of \eqref{eq:utilityinpf_CA1} and \eqref{eq:utilityinpf_CA2}, for any $t, t_k+1\leq t \leq t_{k+1}$ for $t_k\in\overline{\mathcal{T}}$, player $n$'s expected utility conditional on $(X_t,A_t)$ in stage game $G_t(V_{t+1},\psi^*_t,b_t)$ is
\begin{align}
&\mathbb{E}_{\hat\pi_t}\left[U^n_{G_{t}(V_{{t}+1},\psi^*_{t},b_t)}|X_t,A_t\right] \nonumber\\
=&\tilde \phi^n_t(c_t,X^{-n}_t,A_t,\hat\pi_t)+\tilde V^n_t(X^n_t,c_{t_k+1},\hat\pi_t).
\label{eq:utilityinpf}
\end{align}
 when $t_k+1\leq t \leq t_{k+1}$ for $t_k\in\overline{\mathcal{T}}$.

Since the second term in \eqref{eq:utilityinpf} does not depend on the players' strategies, an equilibrium of the stage game $G_t(V_{t+1},\psi^*_t,b_t)$ is also an equilibrium of the game $G'_t(V_{t+1},\psi^*_t,b_t)$ where each player $n\in\mathcal{N}$ has utility
\begin{align}
U^n_{G'_t(V_{t+1},\psi^*_t,b_t)}:=&\tilde \phi^n_t(c_t,X^{-n}_t,A_t,\hat\pi_t).
\label{eq:utilityinpf_simplified}
\end{align}

For any $b_t\in\mathcal{B}_t$, since $\mathcal{A}^n_t$ is a finite set for any $n\in N$ in \textbf{Game M},  the game
$G'_t(V_{t+1},\psi^*_t,b_t)$ has at least one Bayesian Nash equilibrium $\tilde\lambda^*_t(b_t)=\{\tilde\lambda^{*n}_t(b_t) \in \Delta(\mathcal{A}^n_t),n\in\mathcal{N}\}$ (see \cite{fudenberg1991game,osborne1994course}).
Define $\lambda^{*n}_t(x^n_t,b_t):=\tilde \lambda^{*n}_t(b_t)$ for all $x^n_t \in\mathcal{X}^n_t, n\in \mathcal{N}$.
Then, we get a \CI\ strategy $\lambda^{*}_t\in BNE_t(V_{t+1},\psi^*_t)$ so that \eqref{eq:DP:BNE} is satisfied at $t$. Moreover, from \eqref{eq:DP:valueupdate},
\begin{align}
&V^n_t(x^n_t,b_t) \nonumber\\
= &D^n_t(V_{t+1},\lambda^*_t,\psi^*_t)(x^n_t,b_t) \nonumber\\
= & \mathbb{E}^{\lambda^*_t}_{\hat\pi_t}\left[\tilde \phi^n_t(c_t,X^{-n}_t,A_t,\hat\pi_t)+\tilde V^n_t(X^n_t,c_{t_k+1},\hat\pi_t)|x^n_t\right] \nonumber\\
\stackrel{(a)}{=} &\mathbb{E}^{\lambda^*_t}_{\hat\pi_t}\left[\tilde \phi^n_t(c_t,X^{-n}_t,A_t,\hat\pi_t)\right]+ \tilde V^n_t(x^n_t,c_{t_{k}+1},\hat\pi_t) \nonumber\\
=: &\tilde U^n_t(c_t,\hat\pi_t) + \tilde V^n_t(x^n_t,c_{t_{k}+1},\hat\pi_t)
\end{align}
where (a) is true because $A^n_t$ only depends on $b_t$ using $\lambda^{*n}_t$, and $X^{-n}_t$ and $X^n_t$ are independent under $\hat\pi_t$.
Then \eqref{eq:valuedecomp} is satisfied at $t$, and the the proof of the claim is complete.

%the optimization problem \eqref{eq:optimalS} at $t$ is equivalent to a game $G_{t}(c_t,\hat\pi_{t})$ defined as follows. In $G_{t}(c_t,\hat\pi_{t})$, there are $N$ players.
%Each player $n \in \mathcal{N}$ observes $X^n_t$, and chooses action $a^n_t$ to maximize the utility given by
%\begin{align}
% \argmax_{a^n_t \in \mathcal{A}^n_t } \left\{
% \mathbb{E}^{\lambda^{-n}_t}_{\hat\pi_{t}} \left[\tilde \phi^n_t(c_t,a^n_t,X^{-n}_t,A^{-n}_t)|\hat\pi_t\right] 
%+ \tilde V^n_t(x^n_t,c^p_t,\hat\pi_t)\right\}
% \label{eq:gamecpi}
%\end{align}
%where $\lambda^{-n}_t$ indicates the strategies of other players.

%Since the second term in \eqref{eq:gamecpi} does not depend on the strategies of the users, an equilibrium of $G_{t}(c_t,\hat\pi_{t})$ is also an equilibrium of the game $G'_{t}(c_t,\hat\pi_{t})$ with the following utility
%\begin{align}
% \argmax_{a^n_t \in \mathcal{A}^n_t } \left\{
% \mathbb{E}^{\lambda^{-n}_t}_{\hat\pi_{t}} \left[\tilde \phi^n_t(c_t,a^n_t,X^{-n}_t,A^{-n}_t)|\hat\pi_t\right] 
% \right\}
% \label{eq:gamecpip}
%\end{align}
%Using the fact that $\mathcal{A}^n_t$ is a finite set for any $n\in N$ in Game \textbf{G2},
%$G'_{t}(c_t,\hat\pi_{t})$ has at least one Bayesian Nash equilibrium (from the results of finite games \cite{osborne1994course}).
%Let $\{\tilde \lambda^{*n}_t(c_t,\hat\pi_{t}), n\in \mathcal{N} \}$ denote a Bayesian Nash equilibrium of $G'_{t}(c_t,\hat\pi_{t})$. 
%Define $\{\lambda^{*n}_t(x^n_t,c_t,\hat\pi_{t}):=\tilde \lambda^{*n}_t(c_t,\hat\pi_{t}), n\in \mathcal{N} \}$.
%Then $\lambda^{*}_t$ solves \eqref{eq:optimalS} for any time $t\in\mathcal{T}$.

As a result of the claim, we obtain a \CI\ strategy profile $\lambda^*$ and a \CI\ update rule $\psi^*$ such that \eqref{eq:DP:end}, \eqref{eq:DP:BNE} and \eqref{eq:DP:valueupdate} are satisfied. It remains to show the consistency \eqref{eq:DP:consistent}.
Using the dynamics of \textbf{Game M} and the fact that $\lambda^{*n}_t(x^n_t,b_t):=\tilde \lambda^{*n}_t(b_t)$, we obtain
\begin{align}
&\frac{\sum_{x^n_{t} \in \mathcal{X}^n_t}p^n_t(x^n_{t+1};x^n_t,a_t) \eta^n_t(x^n_{t},y^n_t,a_t,b_t)\pi^n_t(x^n_t)}
	{\sum_{x'^n_{t} \in \mathcal{X}^n_t}\eta^n_t(x'^n_{t},y^n_t,a_t,b_t)\pi^n_t(x'^n_t)} 
\nonumber\\
%=&\frac{\sum_{x^n_{t} \in \mathcal{X}^n_t}p^n_t(x^n_{t+1};x^n_t,a_t) q^n_t(y^n_{t};x^n_t,a_t)\lambda^{*n}_t(x^n_t,c_t,\hat\pi_{t})(a^n_t) \hat\pi^{n}_{t}(x^n_{t})}
%	{\sum_{x'^n_{t} \in \mathcal{X}^n_t}q^n_t(y^n_{t};x'^n_t,a_t)\lambda^{*n}_t(x'^n_t,c_t,\hat\pi_{t})(a^n_t) \hat\pi^{n}_{t}(x'^n_{t})} 
%\nonumber\\
=&\frac{\sum_{x^n_{t} \in \mathcal{X}^n_t}p^n_t(x^n_{t+1};x^n_t) q^n_t(y^n_{t};x^n_t)\tilde \lambda^{*n}_t(b_t)(a^n_t) \hat\pi^{n}_{t}(x^n_{t})}
	{\sum_{x'^n_{t} \in \mathcal{X}^n_t}q^n_t(y^n_{t};x'^n_t)\tilde \lambda^{*n}_t(b_t)(a^n_t) \hat\pi^{n}_{t}(x'^n_{t})} 
\nonumber\\
=&\frac{\sum_{x^n_{t} \in \mathcal{X}^n_t}p^n_t(x^n_{t+1};x^n_t) q^n_t(y^n_{t};x^n_t) \hat\pi^{n}_{t}(x^n_{t})}
	{\sum_{x'^n_{t} \in \mathcal{X}^n_t}q^n_t(y^n_{t};x'^n_t) \hat\pi^{n}_{t}(x'^n_{t})} 
\nonumber\\
=& \hat\psi^{n}_{t}(y^n_t,\hat\pi_t)(x^n_{t}) = \psi^{*n}_t(y^n_t,a_t,b_t)(x^n_{t}).
\end{align}
Thus, $\psi^*_t$ satisfies \eqref{eq:tilpsiS}, and $\psi^*_t$ is consistent with $\lambda^*_t$. Therefore \eqref{eq:DP:consistent} holds.

Since $(\lambda^*,\psi^*)$ solves the dynamic program defined by \eqref{eq:DP:end}-\eqref{eq:DP:valueupdate}, it is a \CI-PBE according to Theorem \ref{thm:commonDP}.
\end{proof}

\bibliographystyle{ieeetr}
\bibliography{games}

\end{document}